\documentclass[orivec]{llncs}

\usepackage[T1]{fontenc}                 
\usepackage{lmodern}                     
\usepackage[utf8]{inputenc}              
\usepackage[english]{babel}              

\usepackage{amsmath, amsfonts, amssymb}  
\usepackage{mathrsfs}                    
\usepackage[dvipsnames]{xcolor}          
\usepackage[]{todonotes}                   

\usepackage{xspace}                      
\usepackage{subcaption}

\usepackage{float}                       

\usepackage{tikz}                        
\usetikzlibrary{arrows,arrows.meta,shapes,shapes.multipart,decorations,automata,backgrounds,petri,positioning,shadows,matrix,decorations.pathmorphing,fit,positioning,calc,backgrounds}
\tikzstyle{pl}=[place,minimum size=6mm]
\tikzstyle{tr}=[transition,minimum size=5mm]
\usepackage{tikz-qtree}

\usepackage[thmmarks,amsmath]{ntheorem}  
\usepackage[defblank]{paralist}          
\usepackage[inline]{enumitem}            
\usepackage{breqn}                       
\usepackage{algorithm}
\usepackage{algpseudocode}
\usepackage{hyperref}
\usepackage{url}

\newcommand{\funsym}[1]{\ensuremath{\mathtt{#1}}}
\newcommand{\setsym}[1]{\ensuremath{\mathit{#1}}}
\newcommand{\cname}[1]{\ensuremath{\mathtt{#1}}\xspace} 

\newcommand{\I}{\ensuremath{\mathcal{I}}\xspace}
\newcommand{\J}{\ensuremath{\mathcal{J}}\xspace}

\renewcommand{\L}{\ensuremath{\mathcal{L}}\xspace}

\newcommand{\V}{\ensuremath{\mathcal{V}}\xspace}

\makeatletter
\newcommand{\xdasharrow}[2][->,>=angle 90]{\tikz[baseline=-\the\dimexpr\fontdimen22\textfont2\relax]{\node[anchor=south,font=\scriptsize, inner ysep=1.5pt,outer xsep=2.5pt](x){\ensuremath{#2}};\draw[shorten <=3.4pt,shorten >=3.4pt,dashed,#1](x.south west)--(x.south east);}}
\makeatother

\newcommand{\set}[1]{\ensuremath{\left\{#1\right\}}}

\newcommand{\size}[1]{\ensuremath{\left|{#1}\right|}}

\newcommand{\sequence}[1]{\ensuremath{\left\langle #1\right\rangle}}

\newcommand{\tuple}[1]{\ensuremath{\left( #1\right)}}
\newcommand{\tup}[1]{\ensuremath{\left\langle #1\right\rangle}}

\newcommand{\naturals}{\ensuremath{\mathbb{N}}\xspace}

\newcommand{\powerset}[1]{\ensuremath{\mathcal{P}(#1)}\xspace}
\newcommand{\union}{\ensuremath{\cup}\xspace}
\newcommand{\intersect}{\ensuremath{\cap}\xspace}
\newcommand{\Union}{\ensuremath{\bigcup}\xspace}

\newcommand{\setmult}[1]{\ensuremath{#1^\oplus}\xspace}
\newcommand{\setsupp}[1]{\ensuremath{\mathit{supp}\left(#1\right)}}
\newcommand{\setsuppp}[1]{\ensuremath{\mathit{supp}(#1)}}

\newcommand{\emptybag}{\ensuremath{\emptyset}}

\newcommand{\func}[1]{\ensuremath{ #1 }}

\newcommand{\dom}[1]{\textsc{dom}(\func{#1})}

\newcommand{\rng}[1]{\textsc{rng}(\func{#1})}

\newcommand{\concat}{\cdot}
\newcommand{\seq}[1]{\ensuremath{{#1}^{*}}}
\newcommand{\emptysequence}{\ensuremath{\epsilon}}

\newcommand{\length}[1]{\size{#1}}

\newcommand{\transitionsystem}[1]{\ensuremath{\Gamma_{#1}}\xspace}
\newcommand{\states}{\ensuremath{S}\xspace}
\newcommand{\istate}{\ensuremath{s_0}\xspace}
\newcommand{\state}{\ensuremath{s}\xspace}

\newcommand{\renameop}[1]{\ensuremath{\rho_{#1}}}
\newcommand{\rename}[2]{\ensuremath{\renameop{#1}(#2)}}
\newcommand{\hideop}[1]{\ensuremath{\hat{\cname{H}}_{#1}}}
\newcommand{\hide}[2]{\ensuremath{\hideop{#1}(#2)}}

\newcommand{\algequiv}{\ensuremath \equiv_{alg}}

\newcommand{\pn}{Petri net\xspace}

\newcommand{\wfnet}{WF-net\xspace}
\newcommand{\wfnets}{WF-nets\xspace}

\newcommand{\markednet}[2]{\ensuremath{\left(#1, #2\right)}}
\newcommand{\marking}{\ensuremath{m}}

\newcommand{\markings}[1]{\ensuremath{\mathbb{M}\left(#1\right)}}

\newcommand{\places}{\ensuremath{P}}
\newcommand{\place}{\ensuremath{\mathit{p}\xspace}}
\newcommand{\transitions}{\ensuremath{T}}
\newcommand{\transition}{\ensuremath{\mathit{t}\xspace}}

\newcommand{\flow}{\ensuremath{F}}

\newcommand{\pre}[1]{\ensuremath{{}^\bullet}{#1}}
\newcommand{\post}[1]{\ensuremath{{#1}^{\bullet}}}

\newcommand{\firesym}[1]{\ensuremath{[#1\rangle}}
\newcommand{\fire}[3]{\ensuremath{#1 \firesym{#2} #3}}
\newcommand{\pnenabled}[2]{\ensuremath{#1\firesym{#2}}}

\newcommand{\pnreachable}[2]{\ensuremath{\mathcal{R}(#1,#2)}}

\newcommand{\qedboxfull}{\vrule height 5pt width 5pt depth 0pt}
\newcommand{\qedfull}{\hfill{\qedboxfull}}

\newcommand{\tpnid}{t-PNID\xspace}
\newcommand{\tpnids}{t-PNIDs\xspace}

\newcommand{\tjn}{t-JN\xspace}
\newcommand{\tjns}{t-JNs\xspace}

\newcommand{\Id}{\func{Id}\xspace}
\newcommand{\type}{\funsym{type}\xspace}

\newcommand{\invar}[1]{{\setsym{In}({#1})}}
\newcommand{\outvar}[1]{{\setsym{Out}({#1})}}
\newcommand{\var}[1]{{\setsym{Var}({#1})}}
\newcommand{\newvar}[1]{{\setsym{Emit}({#1})}}
\newcommand{\delvar}[1]{{\setsym{Collect}({#1})}}
\newcommand{\inp}{\ensuremath{\mathit{in}}}
\newcommand{\outp}{\ensuremath{\mathit{out}}}

\newcommand{\idset}{\I\xspace}
\newcommand{\varset}{\V\xspace}
\newcommand{\labelset}{\ensuremath{\Lambda}\xspace}
\newcommand{\colset}{\funsym{C}}

\newcommand{\restr}[2]{\left.#1\right|_{#2}}

\newcommand{\proj}[2]{\ensuremath{{#1}_{{\mid {#2}}}}}

\newcommand{\project}[2]{\ensuremath{\pi_{#1}\left(#2\right)}}

\newcommand{\composeOperator}{\ensuremath{\uplus}}
\newcommand{\compose}[2]{\ensuremath{#1 \composeOperator #2}}

\newcommand{\jnsequenceop}{\ensuremath ; }
\newcommand{\jnparallelop}{\ensuremath \parallel}
\newcommand{\jnchoiceop}{\ensuremath + }
\newcommand{\jnselfloopop}{\ensuremath \# }

\newcommand{\jnsequence}[2]{\ensuremath \left( #1 \jnsequenceop #2 \right) }
\newcommand{\jnparallel}[2]{\ensuremath \left( #1 \jnparallelop #2 \right)}
\newcommand{\jnchoice}[2]{\ensuremath \left( #1 \jnchoiceop #2 \right) }
\newcommand{\jnselfloop}[2]{\ensuremath \left( #1 \jnselfloopop #2 \right) }
\newcommand{\jnplace}[2]{\ensuremath \left[#1, #2 \right]}

\newcommand{\jnnewidentifierop}{\ensuremath \triangleleft}
\newcommand{\jnnewidentifier}[4]{\ensuremath \left( #1 \jnnewidentifierop \tuple{#2, #3, #4}\right)}

\makeatletter
\g@addto@macro\normalsize{
    \setlength{\abovecaptionskip}{-2pt}
    \setlength{\belowcaptionskip}{12pt}
    \setlength\abovedisplayskip{3pt}
    \setlength\belowdisplayskip{3pt}
    \setlength\abovedisplayshortskip{3pt}
    \setlength\belowdisplayshortskip{3pt}
}

\newcommand{\figlabel}[1]{\label{fig:#1}}
\newcommand{\tbllabel}[1]{\label{tbl:#1}}
\newcommand{\deflabel}[1]{\label{def:#1}}
\newcommand{\thmlabel}[1]{\label{thm:#1}}

\newcommand{\corlabel}[1]{\label{cor:#1}}

\newcommand{\seclabel}[1]{\label{sec:#1}}

\newcommand{\figref}[1]{Fig.~\ref{fig:#1}}
\newcommand{\tblref}[1]{Tbl.~\ref{tbl:#1}}
\newcommand{\defref}[1]{Def.~\ref{def:#1}}
\newcommand{\thmref}[1]{Thm.~\ref{thm:#1}}

\newcommand{\corref}[1]{Cor.~\ref{cor:#1}}

\newcommand{\secref}[1]{Section~\ref{sec:#1}}

\newcounter{theoremCounter}
\newcounter{lemmaCounter}
\newcounter{definitionCounter}
\newcounter{propositionCounter}
\newcounter{corollaryCounter}
\newcounter{exampleCounter}
\newcounter{remarkCounter}

\newcounter{assumptionCounter}

\theoremseparator{.}
\theorembodyfont{\itshape}
\theoremsymbol{$\triangleleft$}
\newtheorem{theorem}[theoremCounter]{Theorem}
\newtheorem{lemma}[lemmaCounter]{Lemma}
\newtheorem{definition}[definitionCounter]{Definition}
\newtheorem{proposition}[propositionCounter]{Proposition}
\newtheorem{corollary}[corollaryCounter]{Corollary}

\theorembodyfont{\normalfont}
\newtheorem{example}[exampleCounter]{Example}
\newtheorem{remark}[remarkCounter]{Remark}

\theoremheaderfont{\itshape}
\theoremsymbol{}
\renewtheorem{property}[remarkCounter]{Property}

\theoremheaderfont{\bfseries}
\theorembodyfont{\itshape}
\newtheorem{assumption}[assumptionCounter]{Assumption}
\theoremheaderfont{\itshape}

\theoremstyle{plain}
\theoremheaderfont{\itshape}
\theorembodyfont{\normalfont}

\theoremseparator{.}
\theoremsymbol{\qedfull}
\newtheorem*{proof}{Proof}
\qedsymbol{\qedfull}

\makeatletter
\@addtoreset{equation}{property}
\makeatother

\definecolor{blue-violet}{rgb}{0.54, 0.17, 0.89}
\definecolor{cadmiumorange}{rgb}{0.93, 0.53, 0.18}
\definecolor{yellow-green}{rgb}{0.6, 0.8, 0.2}
\definecolor{green1}{rgb}{0.12, 0.3, 0.17}
\definecolor{byzantium}{rgb}{0.44, 0.16, 0.39}

\title{On the Reconstructability and Rediscoverability of Typed Jackson Nets\\(Extended version)}
\author{Daniël Barenholz \inst{1} 
	\and Marco Montali \inst{2}
	\and Artem Polyvyanyy \inst{3}
	\and Hajo A. Reijers \inst{1}
	\and Andrey Rivkin \inst{2,4}
	\and Jan Martijn E. M. van der Werf \inst{1}}
\institute{%
	Department of Information and Computing Sciences, Utrecht University\\
	Princetonplein 5, 3584 CC Utrecht, The Netherlands\\
	\email{\{d.barenholz,h.a.reijers,j.m.e.m.vanderwerf\}@uu.nl}
	\and
	Faculty of Computer Science, Free University of Bozen-Bolzano\\
	piazza Domenicani 3, 39100, Bolzano, Italy\\
	\email{montali@inf.unibz.it}\\
	\and
	The University of Melbourne, Victoria 3010, Australia\\
	\email{artem.polyvyanyy@unimelb.edu.au}
	\and
	Department of Applied Mathematics and Computer Science,\\ Technical University of Denmark\\
	Richard Petersens Plads 321, 2800 Kgs. Lyngby, Denmark\\
	\email{ariv@dtu.dk}
}

\authorrunning{Barenholz, D. et al.}

\begin{document}
	\maketitle
	
	\begin{abstract}
%
%
A process discovery algorithm aims to construct a model from data generated by historical system executions such that the model describes the system well.
Consequently, one desired property of a process discovery algorithm is \emph{rediscoverability}, which ensures that the algorithm can construct a model that is behaviorally equivalent to the original system.
A system often simultaneously executes multiple processes that interact through object manipulations.
This paper presents a framework for developing process discovery algorithms for constructing models that describe interacting processes based on typed Jackson Nets that use identifiers to refer to the objects they manipulate.
Typed Jackson Nets enjoy the \emph{reconstructability} property which states that the composition of the processes and the interactions of a decomposed typed Jackson Net yields a model that is bisimilar to the original system.
We exploit this property to demonstrate that if a process discovery algorithm ensures rediscoverability, the system of interacting processes is rediscoverable.
\end{abstract}
\section{Introduction}
\seclabel{introduction}

Business processes are fundamental to a wide range of systems.
A business process is a collection of activities that, when performed, aims to achieve a business objective at an organization. 
Examples of business processes are an order-to-cash process at a retailer, a medical assessment process at a hospital, or a credit check process at a bank.
Business processes are modeled using process modeling languages, such as Petri nets, and used for communication and analysis purposes~\cite{vanderAalst2000}. 
Petri nets provide a graphical representation of the flow of activities within a process and can be used to model various types of concurrent and sequential behavior~\cite{Murata1989}. 

A process discovery algorithm aims to automatically construct a model from data generated by historical process executions captured in an event log of the system, such that the model describes the system well.
A desired property of a discovery algorithm is \emph{rediscoverability}.
This property states that if a system $S$, expressed as a model $M$, generates an event log $L$, then a discovery algorithm with the rediscoverability property should construct $M$ from $L$. 
In other words, the algorithm can reverse engineer the model of the system from the data the model has generated. 
Only a few existing algorithms guarantee this property.
For example, if the model is a block-structured workflow net, and the event log is directly-follows complete, then the $\alpha$-Miner algorithm~\cite{AalstWM04} can rediscover the net that generated the event log.
Similarly, again under the assumption that the event log is directly-follows complete, Inductive Miner~\cite{Leemans2013} can rediscover process trees without duplicate transitions, self-loops, or silent transitions.

\begin{figure}[t]
\vspace{-2mm}
\centering
\includegraphics[width=\textwidth]{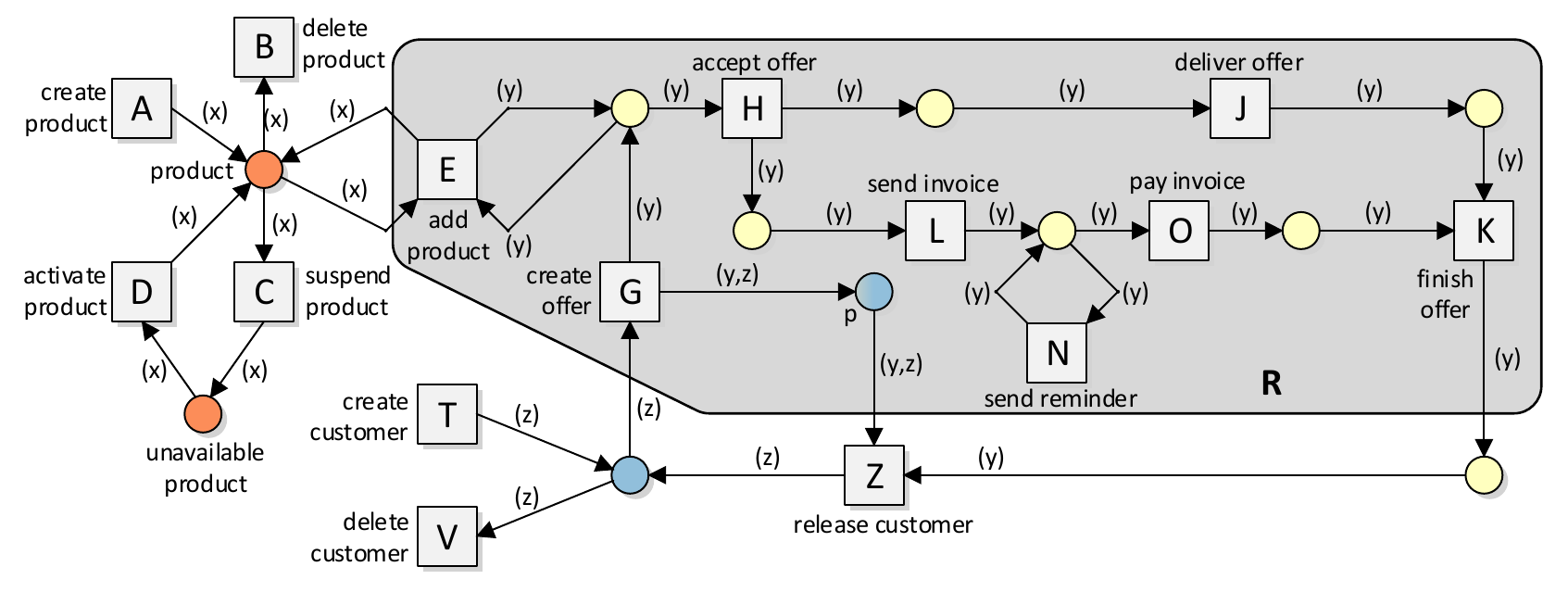}
\vspace{-1mm}
\caption{A retailer system of three interacting processes.}
\figlabel{runningexample}
\vspace{-10mm}
\end{figure}

Most existing process discovery algorithms assume that a system executes a single process~\cite{aalst22_foundations}.
Consequently, an event log is defined as a collection of sequences where a sequence describes the execution of a single process instance.
However, many information systems, such as enterprise resource planning systems, do not satisfy this assumption.
A system often executes multiple interacting processes~\cite{Fahland2019,vanderWerf2022}.
For example, consider a retailer system that executes three processes: an order, product, and customer management process, as depicted in~\figref{runningexample}.
These processes are intertwined.
Specifically, only available products may be ordered, and customers can only have one order at a time.
Consequently, events do not belong to a single process but relate to several processes.
For instance, consider an event $e$ in some event log that occurred as transition $G$ was executed for some customer $c$ and created a new order $o$ in the system.
Event $e$ relates to the customer process instance $c$ and the order process instance $o$.
Traditional process discovery techniques require event $e$ to be stored in multiple event logs and generate multiple models, one for each process~\cite{adams22_extractingfeatures}. 

A different approach is taken in artifact or object-centric process discovery~\cite{aalstB20_discovering,LuNWF15} and agent system discovery~\cite{TourPK21,TourPKS2022agentdiscovery}.
In object-centric process discovery, instead of linking each event to a single object, events can be linked to multiple objects stored in object-centric event logs~\cite{Berti2022}.
Existing object-centric discovery algorithms project the input event log on each object type to create a set of ``flattened'' event logs.
For each event log, a model is discovered, after which these models are combined into a single model~\cite{aalstB20_discovering}.
In general, flattening is lossy~\cite{adams22_extractingfeatures}, as in this step, events can disappear~\cite{aalstB20_discovering}, be duplicated (convergence)~\cite{aalst19_divergence}, or lead to wrong event orders (divergence)~\cite{aalst19_divergence}.
In agent system discovery, instead of interacting objects, a system is viewed as composed of multiple autonomous agents, each driving its processes that interact to achieve an overall objective of the system~\cite{TourPK21}.
An agent system discovery algorithm proceeds by decomposing the input event log into multiple event logs, each composed of events performed by one agent (type) and an event log of interactions, and then discovering agent and interaction models and composing them into the resulting system~\cite{TourPKS2022agentdiscovery}.

\begin{figure}[t]
	\centering
	\begin{tikzpicture}[->,>=stealth',auto,x=10mm,y=1cm,node distance=15mm and 3mm,thick,  every node/.style={scale=1.2}]
		\node[rectangle,draw=lightgray, minimum width=8.5cm, minimum height=2.5cm] (r) at  (3.55,0.3) {};
		\node[rectangle,draw=lightgray, minimum width=8.5cm, minimum height=2.25cm] (r) at (3.55,-2.55) {};
			
		\node[tr, label=center:$M$] (M) {};
		\node[tr, right of = M, label=center:$M_1$, yshift=0.5cm, xshift=1.75cm]  (M1) {}; 
		\node[below of = M1,rotate=90,xshift=10mm] () {$\cdots$};
		\node[tr, right of = M, label=center:$M_n$, yshift=-0.5cm, xshift=1.75cm]  (Mn) {}; 
		\node[tr, right of = M, label=center:$M'$, xshift=5cm]  (M') {};
		
		\node[tr, below of = M, label=center:$L$, yshift=-0.5cm]    (L) {}; 
		\node[tr, right of = L, label=center:$L_1$, yshift=0.5cm]   (L1) {}; 
		\node[below of = L1,rotate=90,xshift=10mm] () {$\cdots$};
		\node[tr, right of = L, label=center:$L_n$, yshift=-0.5cm]  (Ln) {};
		
		\node[tr, right of = L1, label=center:$D_1$, xshift=2cm] (D1) {};
		\node[below of = D1,rotate=90,xshift=10mm] () {$\cdots$};
		\node[tr, right of = Ln, label=center:$D_n$, xshift=2cm] (Dn) {};
		\node[tr, below of = M', label=center:$D'$, yshift=-0.5cm]  (D') {};
		
		\node[right of = M, label=center:{\tiny{project}}, xshift=-0.4cm] () {};
		\node[right of = L, label=center:{\tiny{project}}, xshift=-0.6cm] () {};
		
		\node[left of = M', label=center:{\tiny{compose}},xshift=0.2cm] () {};
		\node[left of = D', label=center:{\tiny{compose}},xshift=0.6cm] () {};
		
		\node[right of = D', label=left:{\tiny{Section 4}}, yshift=3.25cm, xshift=-0.7cm] () {};
		\node[right of = D', label=left:{\tiny{Section 5}}, yshift=-1cm, xshift=-0.7cm] () {};
		
		\node[below of = Dn, yshift=7mm] () {};
		
		\path[->, thin]
			(M) edge (M1)
			(M) edge (Mn)
			
			(L) edge (L1)
			(L) edge (Ln)
			
			(M1) edge (M')
			(Mn) edge (M')
			
			(D1) edge (D')
			(Dn) edge (D')
		;
		
		\path[->, draw=blue-violet]
		(M) edge node[anchor=center, xshift=-0.6cm, yshift=0.2cm]{\textcolor{blue-violet}{\tiny{generates}}} (L);
		
		\path[<->, draw=blue-violet,dotted]
			(M1) edge node[anchor=center, xshift=1.8mm, yshift=-0.5mm]{\textcolor{blue-violet}{\tiny{(?)}}} (L1)
			(Mn) edge node[anchor=center, xshift=0.3mm, yshift=-2.5mm]{\textcolor{blue-violet}{\tiny{(?)}}} (Ln)
		;
		\path[<->, draw=cadmiumorange, dotted]
			(M)  edge[bend left=30] node[anchor=center, yshift=1.5mm]{\textcolor{cadmiumorange}{\tiny{(?)}}}                (M')
			(M') edge               node[anchor=center, yshift=1.8mm, xshift=2.0mm]{\textcolor{cadmiumorange}{\tiny{(?)}}}  (D')
		;
		\path[->, thin]
			(L1) edge node[anchor=center, yshift=1.5mm]{\tiny{discover}} (D1)
			(Ln) edge node[anchor=center, yshift=1.5mm]{\tiny{discover}} (Dn)
		;
		\path[<->, draw=green1, dotted]
			(M1) edge node[anchor=center, yshift=-2mm, xshift=4mm]{\textcolor{green1}{\tiny{(?)}}} (D1)
			(Mn) edge node[anchor=center, xshift=-0.4mm, yshift=-2.5mm]{\textcolor{green1}{\tiny{(?)}}} (Dn)
		;
	\end{tikzpicture}
	\vspace{3mm}
	\caption{The framework for rediscoverability of systems of interacting processes.}
	\figlabel{principleIdea}
	\vspace{-10mm}
\end{figure}
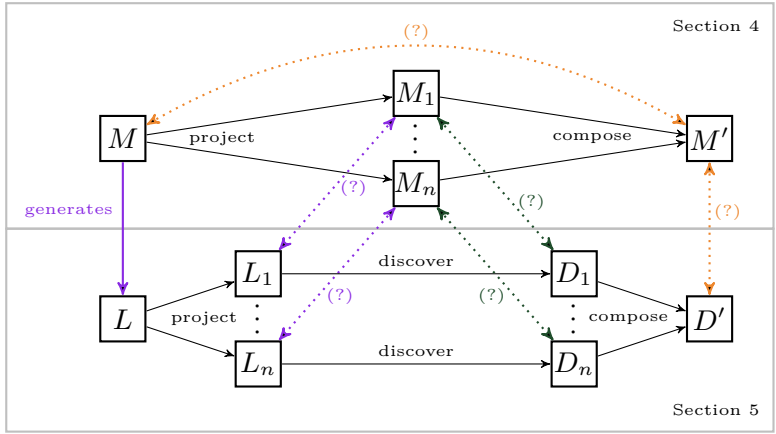
%

In this paper, we study under what conditions projections in event logs can guarantee rediscoverability for interacting processes, represented as typed Jackson Nets, a subclass of typed Petri nets with identifiers~\cite{PolyvyanyyWOB19ISM,vanderWerf2022}.
The class of typed Jackson Nets is inspired by Box Algebra~\cite{Best2002} and Jackson Nets~\cite{vanHee2009}, which are (representations of) block-structured workflow nets that are \emph{sound}~\cite{Aalst1997} by construction~\cite{Leemans2013}. 
As we demonstrate, typed Jackson Nets exhibit a special property: they are \emph{reconstructable}.
Composing the projections of each type is insufficient for reconstructing a typed Jackson Net. 
Instead, if the subset-closed set of all type combinations is considered, the composition returns the original model of the system.
We show how the reconstructability property can be used to develop a framework for rediscoverability of typed Jackson Nets using traditional process discovery algorithms.
The framework builds upon a divide and conquer strategy, as depicted in \figref{principleIdea}. 
The principle idea of this strategy is to project an event log $L$ generated by some model $M$ of the system onto logs $L_1, \ldots, L_n$. 
Then, if these projected event logs satisfy the conditions of a process discovery algorithm, composition of the resulting models $D_1, \ldots, D_n$ into model $D'$ should rediscover the original model of the system.
In this framework, we show that every projected event log is also an event log of the corresponding projected model.
Consequently, if a process discovery algorithm guarantees the rediscoverability of projected models, then the composition operator for typed Jackson Nets can be used to ensure the rediscoverability of the original system.

The next section presents the basic notions. 
In \secref{typed_jackson_nets_section}, we introduce typed Jackson Nets, which, as shown in \secref{decomposability}, are reconstructable.
We define a framework for developing discovery algorithms that guarantee rediscoverability 
in \secref{framework}.
We conclude the paper in \secref{discussion_conclusion}. 


	\section{Preliminaries} \seclabel{preliminaries}

Let $S$ and $T$ be two possibly infinite sets.
The powerset of $S$ is denoted by $\powerset{S} = \set{ S' \mid S' \subseteq S}$ and $\size{S}$ denotes the cardinality of $S$.
Two sets $S$ and $T$ are \emph{disjoint} if $S \cap T = \emptyset$, with $\emptyset$ denoting the empty set.
The cartesian product of two sets $S$ and $T$, is defined by $S \times T = \set{(a,b) \mid a \in S, b \in T}$.
The generalized cartesian product for some set $S$ and and sets $T_s$ for $s \in S$ is defined as $\Pi_{s \in S} T_s = \set{f: S \to \bigcup_{s \in S} T_s \mid \forall s \in S: f(s) \in T_s}$.
Given a relation $R \subseteq S \times T$, its range is defined by $\rng{R} = \set{y \in T \mid \exists x \in S : (x,y)\in R}$.
Similarly, the domain of $R$ is defined by $\dom{R} = \set{x \in S \mid \exists y \in T: (x, y) \in R}$.
Restricting the domain of a relation to a set $U$ is defined by $\proj{R}{U} = \{(a,b) \in R \mid a \in U \}$.

A \emph{multiset} $m$ over $S$ is a mapping of the form $m:S\rightarrow \naturals$, where $\naturals = \{0, 1, 2, \ldots\}$ denotes the set of natural numbers.
For $s \in S$, $m(s) \in \naturals$ denotes the number of times $s$ appears in multiset $m$. We write $s^n$ if $m(s)=n$.
For $x \not\in S$, $m(x) = 0$.
We use $\setmult{S}$ to denote the set of all finite multisets over $S$ and overload $\emptyset$ to also denote the empty multiset. The size of a multiset is defined by $|m|=\sum_{s\in S}m(s)$.
The support of $m\in\setmult{S}$ is the set of elements that appear in $m$ at least once: $\setsupp{m} = \set{s\in S\mid m(s) > 0}$.
Given two multisets $m_1$ and $m_2$ over $S$:
\begin{inparaenum}[\it (i)]
    \item $m_1 \subseteq m_2$ (resp., $m_1 \subset m_2$) iff $m_1(s) \leq m_2(s)$ (resp., $m_1(s) < m_2(s)$) for each $s \in S$;
    \item $(m_1 + m_2)(s) = m_1(s) + m_2(s)$ for each $s \in S$; and
    \item if $m_1 \subseteq m_2$, $(m_2 - m_1)(s) = m_2(s) - m_1(s)$ for each $s \in S$.
\end{inparaenum}
%

A \emph{sequence} over $S$ of length $n \in \naturals$ is a function $\sigma : \{1,\ldots,n\} \to S$.
If $n > 0$ and $\sigma(i) = a_i$, for $1\leq i \leq n$, we write $\sigma = \sequence{a_1, \ldots, a_n}$.
The length of a sequence $\sigma$ is denoted by $\size \sigma$.
The sequence of length $0$ is called the \emph{empty sequence}, and is denoted by $\emptysequence$.
The set of all finite sequences over $S$ is denoted by $S^*$.
We write $a \in \sigma$ if there is $1 \leq i \leq \size \sigma$ such that $\sigma(i) = a$ and $\setsupp{\sigma} = \{a \in S \mid \exists 1 \leq i \leq \size \sigma : \sigma(i) = a \}$.
\emph{Concatenation} of two sequences $\nu,\gamma \in S^*$, denoted by $\sigma = \nu \concat \gamma$, is a sequence defined by $\sigma : \{ 1, \ldots, \size \nu+\size \gamma\}\rightarrow S$, such that $\sigma(i) = \nu(i)$ for $1 \leq i \leq \size \nu$, and $\sigma(i) = \gamma(i - \size \nu)$ for $\size \nu+1 \leq i \leq \size \nu+\size \gamma$.
Projection of sequences on a set $T$ is defined inductively by $\proj{\emptysequence}{T} = \emptysequence$, $\proj{(\sequence{a}\concat\sigma)}{T} = \sequence{a}\concat\proj{\sigma}{T}$ if $a \in T$ and $\proj{(\sequence{a}\concat\sigma)}{T} = \proj{\sigma}{T}$ otherwise.
Renaming a sequence with an injective function $r : S \rightarrow T$ is defined inductively by $\rho_r(\emptysequence) = \emptysequence$, and  $\rho_r(\sequence{a}\concat\sigma) = \sequence{r(a)}\concat\rho_r(\sigma)$.
Renaming is extended to multisets of sequences as follows: given a multiset $m \in \setmult{(S^*)}$, we define $\rho_r(m) = \sum_{\sigma\in\setsupp{m}} \sigma(m)\cdot\rho_r(\sigma)$.
For example, \smash{$\rho_{\{x\mapsto a, y \mapsto b\}}(\sequence{x,y}^3) = \sequence{a,b}^3$}.

A \emph{directed graph} is a pair $(V,A)$ where $V$ is the set of vertices, and $A \subseteq V \times V$ the set of arcs. 
Two graphs $G_1 = (V_1, A_1)$ and $G_2 = (V_2, A_2)$ are \emph{isomorphic}, denoted by $G_1 \leftrightsquigarrow G_2$, if a bijection $b : V_1 \rightarrow V_2$ exists, such that $(v_1, v_2) \in A_1$ iff $(b(v_1), b(v_2)) \in A_2$.

Given a finite set $A$ of (action) labels, a \emph{(labeled) transition system} (LTS) over $A$ is a tuple $\transitionsystem{A} = (\states,A,\istate, \to)$, where $S$ is the (possibly infinite) set of \emph{states}, $\istate$ is the \emph{initial state} and $\to\ \subset (\states\times (A \cup \{\tau\}) \times \states)$ is the \emph{transition relation}, where $\tau\not\in A$ denotes the silent action~\cite{vanGlabbeek1993}.
In what follows, we write $\state \xrightarrow{a} \state'$ for $(\state,a,\state') \in \to$. 
%
%
Let $r : A \to (A' \cup \{\tau\})$ be an injective, total function.
Renaming $\transitionsystem{}$ with $r$ is defined as $\rename{r}{\transitionsystem{}} = (\states, A \setminus A', \istate, \to')$ with $(\state,r(a),\state') \in \to'$ iff $(\state,a,\state') \in \to$.
Given a set $T$, hiding is defined as $\hide{T}{\transitionsystem{}} = \rename{h}{\transitionsystem{}}$ with $h : A \rightarrow A \cup \{\tau\}$ such that $h(t) = \tau$ if $t \in T$ and $h(t)=t$ otherwise.
Given $a\in A$, $p \xdasharrow{~a~} q$ denotes a \emph{weak transition relation} that is defined as follows:
\begin{inparaenum}[\it (i)]
    \item $p \xdasharrow[->]{~a~} q$ iff $p (\xrightarrow{\tau})^* q_1\xrightarrow{a}q_2 (\xrightarrow{\tau})^* q$;
    \item $p \xdasharrow[->]{\ensuremath{~\tau~}} q$ iff $p (\xrightarrow{\tau})^* q$.
\end{inparaenum}
Here, $(\xrightarrow{\tau})^*$ denotes the reflexive and transitive closure of $\xrightarrow{\tau}$.

    Let $\transitionsystem{1} = (\states_1, A, \state_{01}, \to_1)$ and $\transitionsystem{2} = (\states_2, A, \state_{02},\to_2)$ be two LTSs.
    A relation $R\subseteq(\states_1 \times \states_2)$ is called a \emph{strong simulation}, denoted as $\transitionsystem{1} \prec_R \transitionsystem{2}$, if for every pair $(p,q)\in R$ and $a\in A \cup \{\tau\}$, it holds that if $p\xrightarrow{a}_1 p'$, then there exists $q'\in\states_2$ such that $q \xrightarrow{a}_2 q'$ and $(p',q')\in R$. Relation $R$ is a \emph{weak simulation}, denoted by $\transitionsystem{1} \preccurlyeq_R \transitionsystem{2}$, iff  for every pair $(p,q)\in R$ and $a\in A \cup \{\tau\}$ it holds that if $p\xrightarrow{a}_1 p'$, then $a = \tau$ and $(p', q) \in R$, or there exists $q'\in\states_2$ such that $q \xdasharrow{~a~}_{\hspace{-.5ex} 2}~q'$ and $(p',q')\in R$.
    Relation $R$ is called a strong (weak) \emph{bisimulation}, denoted by $\transitionsystem{1} \sim_R \transitionsystem{2}$ ($\transitionsystem{1} \approx_R \transitionsystem{2}$) if both $\transitionsystem{1} \prec \transitionsystem{2}$ ($\transitionsystem{1} \preccurlyeq_R \transitionsystem{2}$) and $\transitionsystem{2} \prec_{R^{-1}} \transitionsystem{1}$ ($\transitionsystem{2} \preccurlyeq_{R^{-1}} \transitionsystem{1}$).
    Given a strong (weak) (bi)simulation $R$, we say that a state $p\in\states_1$ is strongly (weakly) rooted (bi)similar to $q\in\states_2$, written $p \sim^r_R  q$ (correspondingly, $p \approx^r_R q$), if  $(p,q)\in R$.
    The relation is called \emph{rooted} iff $(s_{01}, s_{02}) \in R$. A rooted relation is indicated with a superscript $^r$.


A weighted \pn is a 4-tuple $(\places,\transitions,\flow, W)$ where
$\places$ and $\transitions$ are two disjoint sets of \emph{places} and \emph{transitions}, respectively, $\flow \subseteq ((\places \times \transitions) \cup (\transitions \times \places))$ is the \emph{flow relation}, and $W : \flow \rightarrow \naturals^+$ is a \emph{weight function}.
For $x\in\places\cup\transitions$, we write $\pre{x}=\set{y\mid (y,x)\in\flow}$ to denote the \emph{preset} of $x$ and $\post{x}=\set{y\mid (x,y)\in\flow}$ to denote the \emph{postset} of $x$.
We lift the notation of preset and postset to sets element-wise.
If for a \pn no weight function is defined, we assume $W(f) = 1$ for all $f \in F$.
A \emph{marking} of $N$ is a multiset $m \in \setmult P$, where $m(p)$ denotes the number of \emph{tokens} in place $\place \in \places$. If $m(\place) > 0$, place $\place$ is called \emph{marked} in marking $m$.
A \emph{marked \pn} is a tuple \markednet{N}{m} with $N$ a weighted \pn with marking $m$.
A transition $\transition\in \transitions$ is enabled in $\markednet{N}{m}$, denoted by $\pnenabled{\markednet{N}{m}}{\transition}$ iff $W((\place,\transition)) \leq m(\place)$ for all $\place \in \pre{\transition}$. An enabled transition can \emph{fire}, resulting in marking $m'$ iff $m'(\place) + W((\place, \transition)) = m(\place) + W((\transition, \place))$, for all $\place\in\places$, and is denoted by $\fire{\markednet{N}{m}}{\transition}{\markednet{N}{m'}}$. We lift the notation of firings to sequences. A sequence $\sigma \in T^*$ is a \emph{firing sequence} iff $\sigma = \emptysequence$, or markings $m_0, \ldots, m_n$ exist such that $\fire{(N,m_{i-1})}{\sigma(i)}{(N,m_{i})}$ for $1 \leq i \leq \size \sigma = n$, and is denoted by $\fire{(N,m_0)}{\sigma}{(N,m_n)}$. If the context is clear, we omit the weighted Petri net $N$.
The set of reachable markings of \markednet{N}{m} is defined by $\pnreachable{N}{m} = \{ m' \mid \exists \sigma \in \transitions^* : \fire{m}{\sigma}{m'} \}$.
The set of all possible finite firing sequences of  \markednet{N}{m} is denoted by $\mathcal{L}(N,m_0) = \{ \sigma \in \seq{T} \mid \fire{m}{\sigma}{m'} \}$.
The semantics of a marked Petri net $(N,m)$ with $N = (\places, \transitions, \flow, W)$ is defined by the LTS $\Gamma_{N,m} = (\setmult \places, T, m_0, \to)$ with $(m,t,m') \in \to$ iff $\fire{m}{t}{m'}$.
A Petri net $N = (P, T, F, W)$ has underlying graph $(P \cup T, F)$.
Two Petri nets $N$ and $N'$ are isomorphic, denoted using $N \leftrightsquigarrow N'$, if their underlying graphs are.

A \emph{workflow net} (WF-net for short) is a tuple $N=(\places, \transitions, \flow,W,\inp,\outp)$ such that:
\begin{inparaenum}[\it (i)]
	\item $(\places, \transitions, \flow, W)$ is a weighted Petri net;
	\item $\inp,\outp\in\places$ are the source and sink place, respectively, with $\pre{\inp} =\post{\outp}= \emptyset$;
	\item every node in $\places \cup \transitions$ is on a directed path from $\inp$ to $\outp$.
\end{inparaenum}
$N$ is called \emph{$k$-sound} for some $k \in \naturals$ iff
\begin{inparaenum}[\it (i)]
	\item it is proper completing, i.e., for all reachable markings $m \in \pnreachable{N}{[\inp^k]}$, if $[\outp^k]\subseteq m$, then $m = [\outp^k]$;
	\item it is weakly terminating, i.e., for any reachable marking $m \in \pnreachable{N}{[\inp^k]}$, the final marking is reachable, i.e., $[\outp^k] \in \pnreachable{N}{m}$; and
	\item it is quasi-live, i.e., for all transitions $t\in\transitions$, there is a marking $m\in\pnreachable{N}{[\inp]}$ such that $\pnenabled{m}{t}$.
\end{inparaenum}
The net is called \emph{sound} if it is $1$-sound.
If it is $k$-sound for all $k\in\naturals$, it is called \emph{generalized sound}~\cite{vanHee2003}.


	\section{Typed Jackson Nets to Model Interacting Processes} \seclabel{typed_jackson_nets_section}
In this section, we introduce typed Jackson Nets as subclass of typed Petri nets with identifiers. 
We show that this class is a natural extension to Jackson Nets, which are representations of block-structured workflow nets.
Typed Jackson Nets are identifier sound and live by construction.

\begin{figure}[t]
	\centering
	\begin{tikzpicture}[->,>=stealth',auto,x=10mm,y=1cm,node distance=11mm and 3mm,thick,  every node/.style={scale=0.8}]
		
		\node[pl,label = below:$p_1$]              (p1) {};
		\node[tr,right of = p1, label=center:$t_1$] (t1) {};
		\node[pl,right of = t1, label=below:$p_2$] (p2) {};
		
		\node[tr,right of = p2, label=center:$t_2$,yshift=1cm,xshift=4mm]  (t2) {};
		\node[tr,right of = p2, label=center:$t_3$,yshift=-1cm,xshift=4mm] (t3) {};
		\node[pl,right of = t2, label=below:$p_3$,yshift=-1cm,xshift=4mm] (p3) {};
		\node[tr,left  of = p3, label=center:$t_4$,xshift=-4mm] (t4) {};
		
		\node[tr,right of = p3, label=center:$t_5$] (t5) {};
		\node[pl,right of = t5, label=below:$p_4$] (p4) {};
		
		\node[below of = t3,yshift=5mm] (a) {};
		
		\path[->]
		(p1) edge (t1)
		(t1) edge (p2)
		
		(p2) edge[bend left  = 20] (t2)
		(t2) edge[bend left = 20] (p3)
		
		(p2) edge[bend right = 20] (t3)
		(t3) edge[bend right = 20] (p3)
		
		(p3) edge (t4)
		(t4) edge (p2)

		(p3) edge (t5)
		(t5) edge (p4)
		;  
	\end{tikzpicture}
	\caption{%
		An example block-structured \wfnet. Each block corresponds to a node in the Jackson type
		$
		\jnsequence{
			p_1
		}{
			\jnsequence{
				t_1
			}{
				\jnsequence{
					\jnselfloop{
						\jnsequence{
							p_2
						}{
							\jnsequence{
								\jnchoice{t_2}{t_3}
							}{
								p_3
							}
						}
					}{
						t_4
					}
				}{
					\jnsequence{
						t_5
					}{
						p_4
					}
				}
			}
		}
		$. As example, the choice between transitions $t_2$ and $t_3$ corresponds to the node %
		$
		\jnsequence{p_2}{\jnsequence{\jnchoice{t_2}{t_3}}{p_3}}
		$
		.
	}
	\label{fig:example_wfnet}
\end{figure}

\subsection{Jackson Nets} \seclabel{jackson_nets}
Whereas \wfnets do not put any restriction on the control flow of activities, block-structured \wfnets divide the control flow in logical blocks~\cite{KoppMWL15_DifferenceGraphBlockStructure}. Each ``block'' represents a single unit of work that can be performed, where this unit of work is either atomic (single transition), or one involving multiple steps (multiple transitions). 
An example block-structured \wfnet is shown in \figref{example_wfnet}.
The main advantage of block-structured \wfnets, is that the block-structure ensures that the \wfnet is sound by definition~\cite{KoppMWL15_DifferenceGraphBlockStructure,Leemans2013,vanHee2009}.
In this paper, we consider Jackson Types and Jackson Nets~\cite{vanHee2009}. 
A Jackson Type is a data structure used to capture all information involved in a single execution of a \wfnet. 


\begin{definition}[Jackson Type~\cite{vanHee2009}]
	The set of \emph{Jackson Types} $\J$ is recursively defined by the following grammar: 
	\begin{align*}
		\J & ::= \mathscr A ^p \mid 
		\jnsequence{\mathscr A^p}{\jnsequence{\J^t}{\mathscr A^p}}
		\\
		\J^t & ::= \mathscr A ^t \mid 
		\jnsequence{\J^t}{\jnsequence{\J^p}{\J^t}}
		\mid \jnchoice{\J^t}{\J^t} \\
		\J^p & ::= \mathscr A ^p \mid 
		\jnsequence{\J^p}{\jnsequence{\J^t}{\J^p}} \mid 
		\jnparallel{\J^p}{\J^p} \mid 
		\jnselfloop{\J^p}{\J^t}
	\end{align*}
where $\mathscr{A} = \mathscr{A}^{p} \cup \mathscr{A}^{t} = \left\{a, b, c, \ldots \right\}$ denotes two disjoint sets of atomic types for places and transitions, resp., and symbols $\jnsequenceop, \jnparallelop, \jnchoiceop, \jnselfloopop$ stand for sequence, parallelism, choices, and loops.
\end{definition}

Multiple Jackson Types may exist for the same \wfnet.
For example, the Jackson Type
$
\jnsequence
{
	\jnsequence{p_1}{t_1}
}
{
	\jnsequence
	{
		\jnselfloop{
			\jnsequence{
				p_2
			}{
				\jnsequence{
					\jnchoice{t_2}{t_3}
				}{
					p_3
				}
			}
		}{
			t_4
		}
	}
	{
		\jnsequence{t_5}{p_4}
	}
}
$
describes the \wfnet of \figref{example_wfnet} as well.
Each net has a unique representation~\cite{vanHee2009}, called its normal form.
We define an algebraic equivalence between types to allow rewriting into the normal form.
\begin{definition}[Algebraic equivalence, normal form~\cite{vanHee2009}] \deflabel{algebraic_equivalence}
The \emph{algebraic\\ equivalence} $\algequiv$ is the smallest equivalence relation on the set of Jackson Types that satisfies the following six rules:
\begin{equation*}\begin{array}{rclcrcl}                                                                                 %
	\jnsequence{\jnsequence{J_0}{J_1}}{J_2} & \algequiv & \jnsequence{J_0}{\jnsequence{J_1}{J_2}} &\  &  %
\jnchoice{\jnchoice{J_0}{J_1}}{J_2} & \algequiv & \jnchoice{J_0}{\jnchoice{J_1}{J_2}} \\                 %
\jnparallel{\jnparallel{J_0}{J_1}}{J_2} & \algequiv & \jnparallel{J_0}{\jnparallel{J_1}{J_2}} &\  &      %
\jnchoice{J_0}{J_1} & \algequiv & \jnchoice{J_1}{J_0} \\                                                 %
\jnparallel{J_0}{J_1} & \algequiv & \jnparallel{J_1}{J_0} & \  &                                          %
\jnselfloop{\jnselfloop{J_0}{J_1}}{J_2} & \algequiv & \jnselfloop{J_0}{\jnselfloop{J_1}{J_2}}            %
\end{array}
\end{equation*}
with $J_0, J_1, J_2 \in \J$ three Jackson Types.

A Jackson Type is in \emph{normal form} iff all brackets are moved to the right using the above rules.
\end{definition}

The class of Jackson Nets is obtained by recursively applying \emph{generation rules}, starting from a singleton net with only one place. 
These generation rules are similar to those defined by Murata~\cite{Murata1989} and preserve soundness~\cite{vanHee2009}. 
Thus, any Jackson Net is sound by construction. 

\begin{definition}[Jackson Net~\cite{vanHee2009}] \deflabel{jackson_net}
	A \wfnet $N = (\places, \transitions, \flow, \inp, \outp)$ is called a \emph{Jackson Net} if it can be generated from a single place $p$ by applying the following five generation rules recursively:
\begin{equation*}
	\begin{array}{llcll}
\mbox{J1:} & p \leftrightarrow \jnsequence{p_1}{\jnsequence{t}{p_2}} & \ \  &
\mbox{J4:} & p \leftrightarrow \jnparallel{p_1}{p_2}\\
\mbox{J2:} & t \leftrightarrow \jnsequence{t_1}{\jnsequence{p_1}{t_2}} & &
\mbox{J5:} & t \leftrightarrow \jnchoice{t_1}{t_2} \\
\mbox{J3:} & p \leftrightarrow \jnselfloop{p}{t} & & & \\
\end{array}
\end{equation*}
We say that $N$ is generated by $p$.
\end{definition}

As shown in~\cite{vanHee2009}, Jackson Nets are completely determined by Jackson Types, and vice versa. 

\begin{theorem}[Jackson Nets and Jackson Types are equivalent~\cite{vanHee2009}]
Let $N_1$ and $N_2$ be two Jackson Nets that are generated by the Jackson Types $J_1$ and $J_2$, resp. Then $N_1$ and $N_2$ are isomorphic iff $J_1 \algequiv J_2$.
\end{theorem}

\subsection{Petri Nets with Identifiers} \seclabel{petrinets_identifiers}
Whereas \wfnets describe all possible executions for a single case, 
systems typically consist of many interacting processes.
The latter can be modeled using typed Petri nets with identifiers (\tpnids for short)~\cite{vanderWerf2022}.
In this formalism, each object is typed and has a unique identifier to be able to refer to it.
Tokens carry vectors of identifiers, which are used to relate objects.
Variables on the arcs are used to manipulate the identifiers.

\begin{definition}[Identifiers, Types and Variables] \deflabel{identifier_types}
	Let \idset, \labelset, and \varset denote countably infinite sets of identifiers, type labels, and variables, respectively. We define:
	\begin{compactitem}
			\item the \emph{domain assignment} function $I : \labelset \rightarrow \powerset{\idset}$, such that $I(\lambda_1)$ is an infinite set, and $I(\lambda_1) \cap I(\lambda_2) \neq \emptyset$ implies $\lambda_1 = \lambda_2$ for all $\lambda_1, \lambda_2 \in \labelset$;
			\item the \emph{id typing} function $\type_{\idset}:\idset\to\labelset$ s.t.\ if $\type_{\idset}(\cname{id})=\lambda$, then $\cname{id}\in I(\lambda)$;
			\item a \emph{variable typing} function $\type_{\varset}:\varset\to\labelset$, prescribing that $x\in\varset$ can be substituted only by values from $I(\type_{\varset}(x))$.
		\end{compactitem}
	When clear from the context, we omit the subscripts of $\type$.
	We lift the $\type$ functions to sets, vectors, and sequences by applying the function on each of its constituents.
\end{definition}

In a \tpnid, each place is annotated with a label, called the \emph{place type}.
A place type is a vector of types, indicating types of identifier tokens the place can carry. Similar to Jackson Types, we use $\jnplace{p}{\lambda}$ to denote that place $p$ has type $\alpha(p) = \lambda$.
Each arc is inscribed with a multiset of vectors of identifiers, such that the type of each variable coincides with the place types.
If the inscription is empty or contains a single element, we omit the brackets. 

\begin{definition}[Typed Petri net with identifiers] \deflabel{tpnids}
	A \emph{typed Petri net with identifiers} (\tpnid) $N$ is a tuple $(\places,\transitions,\flow,\alpha,\beta)$, where:
	\begin{compactitem}
			\item $(\places,\transitions,\flow)$ is a classical Petri net;
			\item $\alpha: \places \to \labelset^*$ is the \emph{place typing function};
			\item $\beta : \flow \to \setmult{(\varset^*)}$ defines for each arc a multiset of \emph{variable vectors} s.t.\ $\alpha(p) = \type(x)$ for any $x \in \setsupp{\beta((p,t))}$ and $\type(y)=\alpha(p')$ for any $y \in \setsupp{\beta((t,p'))}$ where $t\in\transitions$, $p\in\pre{t}$, $p'\in\post{t}$.
		\end{compactitem}
\end{definition}


A marking of a \tpnid is the configuration of tokens over the set of places.
Each token in a place should be of the correct type, i.e., the vector of identifiers carried by a token in a place should match the corresponding place type.
The set $\colset(\place)$ defines all possible vectors of identifiers a place $p$ may carry.

\begin{definition}[Marking] \deflabel{marking}
	Given a \tpnid $N = (\places, \transitions, \flow, \alpha, \beta)$, and place $\place \in \places$, its \emph{id set} is $\colset(\place) = \prod_{1 \leq i \leq |\alpha(\place)|} I(\alpha(\place)(i))$.
	A \emph{marking} is a function $\marking \in \markings{N}$, with $\markings{N} = \places \to \setmult{(\seq{\labelset})}$, such that $\marking(\place) \in \setmult{\colset(p)}$, for each place $\place \in \places$. The set of identifiers used in $\marking$ is denoted by $\Id(\marking) = \Union_{\place \in \places} \rng{\setsupp{\marking(\place)}}$
	The pair \markednet{N}{\marking} is called a \emph{marked \tpnid}.
\end{definition}

To define the semantics of a \tpnid, the variables need to be valuated with identifiers.

\begin{definition}[Variable sets~\cite{vanderWerf2022}] \deflabel{variable_sets}
	Given a \tpnid $N=(\places,\transitions,\flow,\alpha,\beta)$, $t\in \transitions$ and $\lambda \in \Lambda$, we define the following sets of variables:
	\begin{compactitem}
		\item \emph{input variables} as $\smash{\invar{t} = \bigcup_{x \in \beta((p,t)), p\in\pre{t}} \rng{\setsupp{x}}}$;
		\item \emph{output variables} as $\smash{\outvar{t} = \bigcup_{x \in \beta((t,p)),p \in\post{t}} \rng{\setsupp{x}}}$;
		\item \emph{variables} as $\var{t} = \invar{t} \cup \outvar{t}$;
		\item \emph{emitting variables} as $\newvar{t} = \outvar{t}\setminus\invar{t}$;
		\item \emph{collecting variables} as $\delvar{t} = \invar{t} \setminus \outvar{t}$;
		\item \emph{emitting transitions} as $E_N(\lambda) = \{ t \mid \exists x \in \newvar{t} \wedge \type(x) = \lambda \}$;
		\item \emph{collecting transitions} as $C_N(\lambda) = \{ t \mid \exists x \in \delvar{t} \wedge \type(x) = \lambda \}$;
		\item \emph{types in $N$} as \smash{$\type(N) = \{ \vec\lambda \mid \exists p \in P : \vec\lambda\in\alpha(p)\}$}.
	\end{compactitem}
\end{definition}

A valuation of variables to identifiers is called a \emph{binding}.
Bindings are used to inject new fresh data into the net via variables that emit identifiers, i.e., via variables that appear only on the output arcs of that transition.
Note that in this definition, freshness of identifiers is local to the marking, i.e., disappeared identifiers (those fully removed from the net through collecting transitions) may be reused, as it does not hamper the semantics of the \tpnid.

\begin{definition}[Firing rule for \tpnids]
 	Given a marked \tpnid $(N,m)$ with $N=(\places,\transitions,\flow,\alpha,\beta)$, a \emph{binding} for  transition $t\in T$ is an injective function  $\psi:\V\rightarrow \I$ such that
 	$\type(v) = \type(\psi(v))$
 	and
 	$\psi(v)\not\in \Id(m)$ iff $v\in\newvar{t}$.
 	Transition $t$ is \emph{enabled} in $(N,m)$ under binding $\psi$, denoted by $\pnenabled{(N,m)}{t,\psi}$ iff $\rho_\psi(\beta(p,t)) \leq m(p)$ for all $p\in\pre{t}$.
 	Its firing results in marking $m'$, denoted by $\fire{(N,m)}{t,\psi}{(N, m')}$, such that $m'(p) + \rho_\psi(\beta(p,t)) = m(p) + \rho_\psi(\beta(t,p))$. 
\end{definition}
The firing rule is inductively extended to sequences.
A marking $m'$ is \emph{reachable} from $m$ if there exists $\eta\in (\transitions\times(\V\rightarrow \I))^*$ such that $\fire{(N,m)}{\eta}{(N,m')}$.
We denote with $\pnreachable{N}{m}$ the set of all markings reachable from $m$ for $\markednet{N}{m}$.
We use $\L\markednet{N}{m}$ to denote all possible firing sequences of \markednet{N}{m}, i.e., $\L\markednet{N}{m} = \{\eta \mid \pnenabled{(N,m)}{\eta}{}\}$ and $\Id(\eta) = \bigcup_{(t,\psi) \in \eta} \rng{\psi}$ for the set of identifiers used in $\eta$.
The execution semantics of a \tpnid is defined as an LTS that accounts for all possible executions starting from a given initial marking.
We say two \tpnids are bisimilar if their induced transition systems are.

 \begin{definition}
	     Given a marked \tpnid $(N,m_0 )$ with $N=(P,T,F,\alpha,\beta)$, its induced transition system  is $\transitionsystem{N,m_0} = (\mathbb{M}(N),(T\times(\V\to\I)),m_0 , \to)$ with $m\xrightarrow{(t,\psi)} m'$ iff  $\fire{\markednet{N}{m}}{t,\psi}{\markednet{N}{m'}}$.
\end{definition}

Soundness properties for \wfnets typically consist of proper completion, weak termination, and quasi-liveness~\cite{AalstHHS11_soundness}. 
Extending soundness to \tpnids gives \emph{identifier soundness}~\cite{vanderWerf2022}.
In \tpnids, each object of a given type ``enters'' the system through an emitting transition, binding it to a unique identifier.
Identifier soundness intuitively states that it should always be possible to remove objects (weak type termination), and that once a collecting transition fires for an object, there should be no remaining tokens referring to the removed object (proper type completion). 


\begin{definition}[Identifier Soundness~\cite{vanderWerf2022}]
	Let $\markednet{N}{m_0}$ a marked \tpnid and $\lambda \in \Lambda$ some type.
	$\markednet{N}{m_0}$ is \emph{$\lambda$-sound} iff it is
	\begin{compactitem}
		\item Proper $\lambda$-completing, i.e., for all $t \in C_N(\lambda)$, bindings $\psi : \V\to\I$ and markings $m, m'\in\pnreachable{N}{m_0}$, if $\fire{m}{t,\psi}{m'}$, then for all identifiers $\cname{id} \in \rng{\restr{\psi}{\delvar{t}}} \cap \Id(m)$ and $\type(\cname{id})=\lambda$, it holds that $\cname{id}\not\in\Id(m')$ \footnote{Here, we constrain $\psi$ only to objects of type $\lambda$ that are only consumed.};
		\item Weakly $\lambda$-terminating, i.e.,  for every $m \in \pnreachable{N}{m_0}$ and identifier $\cname{id} \in I(\lambda)$ such that $\cname{id} \in \Id(m)$, there exists a marking $m' \in \pnreachable{N}{m}$ with $\cname{id} \not\in \Id(m')$.
	\end{compactitem}
	If it is $\lambda$-sound for all $\lambda\in \type(N)$, then it is \emph{identifier sound}.
\end{definition}

\subsection{Typed Jackson Nets} \seclabel{typed_jackson_nets}
In general, identifier soundness is undecidable for \tpnids~\cite{vanderWerf2022}.
Similar as Jackson Nets restrict \wfnets to blocks, \emph{typed Jackson Nets} (\tjns) restrict \tpnids to blocks, while guaranteeing identifier soundness and liveness. 
For \tjns, we disallow multiplicity on arcs and variables, i.e., $\beta(f)(v) \le 1$ for all $f \in F$ and $v \in \V$, and imply a bijection on variables and identifier types.
This prevents place types like $\lambda = \sequence{x, x}$.
Assuming a G\"odel-like number on types (cf.~\cite{vanHee2009}), place types and arc inscriptions can be represented as sets.
Similar as Jackson Types describe Jackson Nets, we apply a notation based on Jackson Types to denote typed Jackson Nets.


\begin{definition}[Typed Jackson Net] \deflabel{typed_jackson_net}
	A \tpnid $N$ is a \emph{typed Jackson Net} if it can be generated from a set of transitions $T'$ by applying any of the following six generation rules recursively. If $N$ is generated from a singleton set of transitions (i.e., $\length{T'} = 1$), $N$ is called \emph{atomic}. 
	
	\begin{enumerate}[label=R\arabic*, wide = 0pt, leftmargin=\parindent]
		\item \label{itm:tjn_R1} Place Expansion:
		$
		\jnplace{p}{\lambda} \leftrightarrow
		\jnsequence{
			\jnplace{p_1}{\lambda}
		}{
			\jnsequence{
				t_1
			}{
				\jnplace{p_2}{\lambda}
			}
		}$
		\begin{center}
			\begin{tikzpicture}[->,>=stealth',auto,x=17mm,y=1.3cm,node distance=15mm and 15mm,thick,  every node/.style={scale=0.7}]
				
				\node[pl,label = below:$p$] (p){};
				\node[left of = p,yshift=8mm](in1){};
				\node[left of = p,yshift=-8mm](in2){};
				\node[right of = p,yshift=8mm](out1){};
				\node[right of = p,yshift=-8mm](out2){};
				
				\path[->] 
				(in1) edge node[above]{$\nu$} (p)
				(in2) edge node[above]{$\nu$} (p)
				(p) edge node[above]{$\nu$} (out1)
				(p) edge node[above]{$\nu$} (out2)
				;  
				
				\node[pl,right of = p, xshift=2.7cm, label = below:$p_1$] (p1){};
				\node[tr,right of = p1, label = center:$t$] (t){};
				\node[pl,right of = t, label = below:$p_2$] (p2){};
				\node[left of = p1,yshift=8mm](in1){};
				\node[left of = p1,yshift=-8mm](in2){};
				\node[right of = p2,yshift=8mm](out1){};
				\node[right of = p2,yshift=-8mm](out2){};
				\path[->] 
				(in1) edge node[above]{$\nu$} (p1)
				(in2) edge node[above]{$\nu$} (p1)
				(p1) edge node[above]{$\mu$} (t)
				(t) edge node[above]{$\mu$} (p2)
				(p2) edge node[above]{$\nu$} (out1)
				(p2) edge node[above]{$\nu$} (out2)
				;  
				
				\draw[{Triangle[open]}-{Triangle[open]}] ([xshift=7.5mm]p.east) -- ([xshift=-7.5mm]p1.west) ;
				
			\end{tikzpicture}
		\end{center}
	
		\item \label{itm:tjn_R2} Transition Expansion:
		$
		t \leftrightarrow
		\jnsequence{
			t_1
		}{
			\jnsequence{
				\jnplace{p}{\lambda}
			}{
				t_2
			}
		}
		$%
		, with $\var{t} \subseteq \lambda$
		\begin{center}
			\begin{tikzpicture}[->,>=stealth',auto,x=17mm,y=1.3cm,node distance=15mm and 15mm, thick,  every node/.style={scale=0.7}]
				\node[tr,label = center:$t$] (t){};
				\node[left of = t,yshift=8mm](in1){};
				\node[left of = t,yshift=-8mm](in2){};
				\node[right of = t,yshift=8mm](out1){};
				\node[right of = t,yshift=-8mm](out2){};
				
				\path[->] 
				(in1) edge node[above]{$\nu_1$} (t)
				(in2) edge node[above]{$\nu_2$} (t)
				(t) edge node[above]{$\nu_3$} (out1)
				(t) edge node[above]{$\nu_4$} (out2)
				;  
				
				\node[tr,right of = t, xshift=2.7cm, label = center:$t_1$] (t1){};
				\node[pl,right of = t1, label = below:$p$] (p){};
				\node[tr,right of = p, label = center:$t_2$] (t2){};
				\node[left of = t1,yshift=8mm](in1){};
				\node[left of = t1,yshift=-8mm](in2){};
				\node[right of = t2,yshift=8mm](out1){};
				\node[right of = t2,yshift=-8mm](out2){};
				\path[->] 
				(in1) edge node[above]{$\nu_1$} (t1)
				(in2) edge node[above]{$\nu_2$} (t1)
				(t1) edge node[above]{$\mu$} (p)
				(p) edge node[above]{$\mu$} (t2)
				(t2) edge node[above]{$\nu_3$} (out1)
				(t2) edge node[above]{$\nu_4$} (out2)
				;  
				
				\draw[{Triangle[open]}-{Triangle[open]}] ([xshift=7.5mm]t.east) -- ([xshift=-7.5mm]t1.west) ;		
			\end{tikzpicture}
		\end{center}
	
		\item \label{itm:tjn_R3} Place Duplication:
		$
		\jnsequence{t_1}{\jnsequence{\jnplace{p}{\lambda}}{t_2}} \leftrightarrow
		\jnsequence{t_1}{\jnsequence{\jnparallel{\jnplace{p}{\lambda}}{\jnplace{p'}{\lambda'}}}{t_2}}
		$%
		, \\with $\lambda' \cap \newvar{\post{p}} = \emptyset$\\ 
		\begin{center}
			\begin{tikzpicture}[->,>=stealth',auto,x=17mm,y=1.3cm,node distance=15mm and 15mm, thick,  every node/.style={scale=0.7}]
				\node[tr, label = center:$t_1$] (t1){};
				\node[pl,right of = t1, label = below:$p$] (p){};
				\node[tr,right of = p, label = center:$t_2$] (t2){};
				\node[left of = t1,yshift=8mm](in1){};
				\node[left of = t1,yshift=-8mm](in2){};
				\node[right of = t2,yshift=8mm](out1){};
				\node[right of = t2,yshift=-8mm](out2){};
				\path[->] 
				(in1) edge node[above]{$\nu_1$} (t1)
				(in2) edge node[above]{$\nu_2$} (t1)
				(t1) edge node[above]{$\mu_1$} (p)
				(p) edge node[above]{$\mu_2$} (t2)
				(t2) edge node[above]{$\nu_3$} (out1)
				(t2) edge node[above]{$\nu_4$} (out2)
				;

				\node[tr,right of = t2, xshift=2.7cm, label = center:$t_1$] (tt1){};
				\node[pl,right of = tt1, label = below:$p$,yshift=8mm] (p1){};
				\node[pl,right of = tt1, label = below:$p'$,yshift=-8mm] (p2){};
				\node[tr,right of = p1, label = center:$t_2$,yshift=-8mm] (tt2){};
				\node[left of = tt1,yshift=8mm](in1){};
				\node[left of = tt1,yshift=-8mm](in2){};
				\node[right of = tt2,yshift=8mm](out1){};
				\node[right of = tt2,yshift=-8mm](out2){};
				\path[->] 
				(in1) edge node[above]{$\nu_1$} (tt1)
				(in2) edge node[above]{$\nu_2$} (tt1)
				(tt1) edge node[above]{$\mu_1$} (p1)
				(p1) edge node[above]{$\mu_2$} (tt2)
				(tt1) edge node[above]{$\mu_3$} (p2)
				(p2) edge node[above]{$\mu_4$} (tt2)
				(tt2) edge node[above]{$\nu_3$} (out1)
				(tt2) edge node[above]{$\nu_4$} (out2)
				;  
				
				\draw[{Triangle[open]}-{Triangle[open]}] ([xshift=7.5mm]t2.east) -- ([xshift=-5.5mm]tt1.west) ;		
			\end{tikzpicture}
		\end{center}
	
		\item \label{itm:tjn_R4} Transition Duplication:
		$
		t \leftrightarrow 
		\jnchoice{t}{t'}
		$
		\begin{center}
			\begin{tikzpicture}[->,>=stealth',auto,x=17mm,y=1.3cm,node distance=15mm and 15mm, thick,  every node/.style={scale=0.7}]
				\node[pl] (p1){};
				\node[tr,right of = p1, label = center:$t$] (t){};
				\node[pl,right of = t, ] (p2){};
				\node[left of = p1,yshift=8mm](in1){};
				\node[left of = p1,yshift=-8mm](in2){};
				\node[right of = p2,yshift=8mm](out1){};
				\node[right of = p2,yshift=-8mm](out2){};
				\path[->] 
				(in1) edge node[above]{$\nu_1$} (p1)
				(in2) edge node[above]{$\nu_2$} (p1)
				(p1) edge node[above]{$\mu_1$} (t)
				(t) edge node[above]{$\mu_2$} (p2)
				(p2) edge node[above]{$\nu_3$} (out1)
				(p2) edge node[above]{$\nu_4$} (out2)
				;  
				
				\node[pl,right of = p2, xshift=2.7cm,] (pp1){};
				\node[tr,right of = pp1, label = center:$t$,yshift=8mm] (t1){};
				\node[tr,right of = pp1, label = center:$t'$,yshift=-8mm] (t2){};
				\node[pl,right of = t1,yshift=-8mm ] (pp2){};
				\node[left of = pp1,yshift=8mm](in1){};
				\node[left of = pp1,yshift=-8mm](in2){};
				\node[right of = pp2,yshift=8mm](out1){};
				\node[right of = pp2,yshift=-8mm](out2){};
				\path[->] 
				(in1) edge node[above]{$\nu_1$} (pp1)
				(in2) edge node[above]{$\nu_2$} (pp1)
				(pp1) edge node[above]{$\mu_1$} (t1)
				(t1) edge node[above]{$\mu_2$} (pp2)
				(pp1) edge node[above]{$\mu_1$} (t2)
				(t2) edge node[above]{$\mu_2$} (pp2)
				(pp2) edge node[above]{$\nu_3$} (out1)
				(pp2) edge node[above]{$\nu_4$} (out2)
				;  
				\draw[{Triangle[open]}-{Triangle[open]}] ([xshift=7.5mm]p2.east) -- ([xshift=-7.5mm]pp1.west) ;
			\end{tikzpicture}
		\end{center}
	
		\item \label{itm:tjn_R5} Self Loop Addition:
		$
		\jnplace{p}{\lambda} \leftrightarrow
		\jnselfloop{
			\jnplace{p}{\lambda}
		}{
			t
		}
		$
		\begin{center}
			\begin{tikzpicture}[->,>=stealth',auto,x=17mm,y=1.3cm,node distance=15mm and 15mm,thick,  every node/.style={scale=0.7}]
				
				\node[pl,label = below:$p$] (p){};
				\node[left of = p,yshift=8mm](in1){};
				\node[left of = p,yshift=-8mm](in2){};
				\node[right of = p,yshift=8mm](out1){};
				\node[right of = p,yshift=-8mm](out2){};
				
				\path[->] 
				(in1) edge node[below]{$\nu$} (p)
				(in2) edge node[above]{$\nu$} (p)
				(p) edge node[below]{$\nu$} (out1)
				(p) edge node[above]{$\nu$} (out2)
				;  
				
				\node[pl,right of = p, xshift=2.7cm, label = below:$p$] (p1){};
				\node[tr,above of = p1, label = center:$t$,yshift=2mm] (t){};
				\node[left of = p1,yshift=8mm](in1){};
				\node[left of = p1,yshift=-8mm](in2){};
				\node[right of = p1,yshift=8mm](out1){};
				\node[right of = p1,yshift=-8mm](out2){};
				\path[->] 
				(in1) edge node[below]{$\nu$} (p1)
				(in2) edge node[above]{$\nu$} (p1)
				(p1) edge[bend left=20] node[left]{$\mu$} (t)
				(t) edge[bend left=20] node[right]{$\mu$} (p1)
				(p1) edge node[below]{$\nu$} (out1)
				(p1) edge node[above]{$\nu$} (out2)
				;  
				
				\draw[{Triangle[open]}-{Triangle[open]}] ([xshift=7.5mm]p.east) -- ([xshift=-7.5mm]p1.west) ;
				
			\end{tikzpicture}
		\end{center}
	
		\item \label{itm:tjn_R6} Identifier Introduction:
		$t \leftrightarrow
		\jnnewidentifier{t}{N_1}{\jnplace{p}{\lambda}}{N_2}
		$%
		, with
		$
		\jnsequence{N_1}{\jnsequence{\jnplace{p}{\lambda}}{N_2}}
		$ a \tjn and $\lambda \cap \var{t} = \emptyset$
		\begin{center}
			\begin{tikzpicture}[->,>=stealth',auto,x=17mm,y=1.3cm,node distance=15mm and 15mm,thick,  every node/.style={scale=0.7}]
				\node[tr,label = center:$t$] (t){};
				\node[left of = t,yshift=8mm](in1){};
				\node[left of = t,yshift=-8mm](in2){};
				\node[right of = t,yshift=8mm](out1){};
				\node[right of = t,yshift=-8mm](out2){};
				
				\path[->] 
				(in1) edge node[below]{$\nu_1$} (t)
				(in2) edge node[above]{$\nu_2$} (t)
				(t) edge node[below]{$\nu_3$} (out1)
				(t) edge node[above]{$\nu_4$} (out2)
				;  
				
				\node[tr,right of = p, xshift=2.7cm, label = center:$t$] (t1){};
				\node[pl,above of = t1, label = above:$p$,yshift=2mm] (p1){};
				\node[tr,left of = p1,label = center:$t_1$] (t2){};
				\node[tr,right of = p1,label = center:$t_2$] (t3){};
				\node[left of = t1,yshift=8mm](in1){};
				\node[left of = t1,yshift=-8mm](in2){};
				\node[right of = t1,yshift=8mm](out1){};
				\node[right of = t1,yshift=-8mm](out2){};
				\path[->] 
				(in1) edge node[below]{$\nu_1$} (t1)
				(in2) edge node[above]{$\nu_2$} (t1)
				(t1) edge[bend left=20] node[left]{$\mu$} (p1)
				(p1) edge[bend left=20] node[right]{$\mu$} (t1)
				(t2) edge node[above]{$\mu$} (p1)
				(p1) edge node[above]{$\mu$} (t3)
				(t1) edge node[below]{$\nu_3$} (out1)
				(t1) edge node[above]{$\nu_4$} (out2)
				;  
				
				\draw[{Triangle[open]}-{Triangle[open]}] ([xshift=7.5mm]t.east) -- ([xshift=-7.5mm]t1.west) ;
			\end{tikzpicture}
		\end{center}
	\end{enumerate}
\end{definition}

An example t-JN is given in \figref{runningexample}. 
Starting with the product process, transitions $C$ and $D$ can be reduced using rule $R2$. 
The resulting transition is a self-loop transition, and can be reduced using $R5$, resulting in the block $\jnnewidentifier{E}{A}{\mathit{product}}{B}$.
This block can be reduced using $R6$, leaving transition $E$.
Transition $E$ is again a self-loop, and can be reduced using $R5$.
The block containing transitions $H$, $J$, $L$ $O$, $N$ and $K$ can be reduced to a single place by applying rules $R1$, $R2$ and $R5$ repeatedly. 
The remaining place is a duplicate place with respect to place $p$, and can be reduced using $R3$.
Applying $R2$ on $G$ and $Z$ results in the block $\jnnewidentifier{G}{T}{\mathit{customer}}{V}$, 
which can be reduced to the transition $G$. 
Hence, the net in \figref{runningexample} is an atomic \tjn.

\begin{theorem}[Identifier Soundness of typed Jackson Nets \cite{vanderWerf2022}] \thmlabel{tjn-is-identifier-sound}
	Given a \tjn $N$, then $N$ is \emph{identifier sound} and \emph{live}.
\end{theorem}

	\section{Decomposability of t-JNs} \seclabel{decomposability}
%
%
%


\tpnids specify a class of nets with explicitly defined interactions between objects of different types within one system.
However, sometimes one may want to focus only on some behaviors exhibited by a given set of object types, by extracting a corresponding net from the original \tpnid model. 
We formalize this idea below.
\begin{definition}[Type projection] \deflabel{type_projection}
Let $N = (P_N, T_N, F_N, \alpha, \beta)$ be a \tpnid and $\Upsilon \subseteq \Lambda$ be a set of identifier types.
The \emph{type projection} of $\Upsilon$ on $N$ is a \tpnid $\project{\Upsilon}{N} = (P_\Upsilon, T_\Upsilon, F_\Upsilon, \alpha_\Upsilon, \beta_\Upsilon)$, where:
\begin{compactitem}
\item $P_\Upsilon = \set{p\in P\mid \Upsilon\subseteq\alpha(p)}$;
\item $T_\Upsilon = \set{t\in T\mid (\pre{t}\cup\post{t})\cap P\neq \emptyset}$;
\item $F_\Upsilon = F\cap ((P_\Upsilon\times T_\Upsilon)\cup (T_\Upsilon\times P_\Upsilon))$;
\item $\alpha_\Upsilon(p)=\Upsilon$, for each $p\in P_\Upsilon$;
\item $\beta_\Upsilon(f)=\restr{\beta(f)}{\type_{\varset}^{-1}(\Upsilon)}$, for each $f\in ((P_\Upsilon\times T_\Upsilon)\cup (T_\Upsilon\times P_\Upsilon))$.
\end{compactitem}
\end{definition}

With the next lemma we explore a property of typed Jackson nets that, in a nutshell, shows that \tjns are closed under the type projection. 
This also indirectly witnesses that \tjns provide a suitable formalism for specifying and manipulating systems with multiple communicating components.

\begin{lemma}
If $N = (P_N, T_N, F_N, \alpha, \beta)$ is a \tjn, then  $\project{\Upsilon}{N}$ is a \tjn as well, for any $\Upsilon \subseteq \type_\Lambda(N)$. 
\end{lemma}

\begin{proof} (sketch) Let us assume for simplicity that $N$ is atomic. Then, using rules from \defref{typed_jackson_net}, $N$ can be reduced to a single transition.  
Starting from this transition, one can construct a \tjn following the net graph construction from \defref{type_projection} using the same rules (but the identifier introduction one), proviso that arc inscriptions are always of type $\Upsilon$.
Then, it is easy to check that the constructed net is indeed the type projection of $\Upsilon$ on $N$.
\end{proof}

%

We define next how \tpnids can be composed and show that \tjns are not closed under the composition.
\begin{definition}[Composition] \deflabel{composition}
	Let $N = (P_N, T_N, F_N, \alpha_N, \beta_N)$ and \\$M = (P_M, T_M, F_M, \alpha_M, \beta_M)$ be two \tpnids. Their \emph{composition} is defined by: $$
	\compose{N}{M} = \left(P_N \cup P_M, T_N \cup T_M, F_N \cup F_M, \alpha_N \cup \alpha_M, \beta_N \cup \beta_M\right)
	$$
\end{definition}

\begin{figure}[t]
  \centering
	\begin{subfigure}{.45\textwidth}
		\centering
		 	\begin{tikzpicture}[->,>=stealth',auto,x=10mm,y=1cm,node distance=11mm and 3mm,thick,  every node/.style={scale=0.7}]
			\node[tr] (a) {$a$};
			\node[pl,right of = a, label=below:$p$] (p) {};
			\node[tr,right of = p] (b) {$c$};
			\node[pl,right of = b, label=below:$q$] (q) {};
			\node[tr,right of = q] (c) {$b$};
			\node[pl,right of = c, label=below:$r$] (r) {};
			\node[tr,right of = r] (d) {$d$};
			\path[->]
				(a) edge node[above,pos=.3]{$x$} (p)
				(p) edge node[above,pos=.3]{$x$} (b)
				(b) edge node[above,pos=.3]{$x$} (q)
				(q) edge node[above,pos=.3]{$x$} (c)
				(c) edge node[above,pos=.3]{$x$} (r)
				(r) edge node[above,pos=.3]{$x$} (d);  
			\end{tikzpicture}
        \caption{\tjn $N$\label{fig:tjn-n}}
	\end{subfigure}
	\begin{subfigure}{.45\textwidth}
		\centering
		\begin{tikzpicture}[->,>=stealth',auto,x=10mm,y=1cm,node distance=11mm and 3mm,thick,  every node/.style={scale=0.7}]
			\node[tr,] (a) {$a$};
			\node[pl,right of = a, label=below:$p$] (p) {};
			\node[tr,right of = p] (b) {$b$};
			\node[pl,right of = b, label=below:$s$] (s) {};
			\node[tr,right of = s] (c) {$c$};
			\node[pl,right of = c, label=below:$r$] (r) {};
			\node[tr,right of = r] (d) {$d$};
			\path[->]
				(a) edge node[above,pos=.3]{$x$} (p)
				(p) edge node[above,pos=.3]{$x$} (b)
				(b) edge node[above,pos=.3]{$x$} (s)
				(s) edge node[above,pos=.3]{$x$} (c)
				(c) edge node[above,pos=.3]{$x$} (r)
				(r) edge node[above,pos=.3]{$x$} (d);  
			\end{tikzpicture}
        \caption{\tjn $M$\label{fig:tjn-m}}
	\end{subfigure}
	\begin{subfigure}{.6\textwidth}
		\centering
			\begin{tikzpicture}[->,>=stealth',auto,x=10mm,y=1cm,node distance=11mm and 3mm,thick,  every node/.style={scale=0.7}]
			\node[tr] (a) {$a$};
			\node[pl,right of = a, label=below:$p$] (p) {};
			\node[tr,right of = p] (b) {$b$};
			\node[pl,right of = b, label=below:$s$,yshift=-1cm] (s) {};
			\node[pl,right of = b, label=below:$q$,yshift=1cm] (q) {};
			\node[tr,right of = s, yshift=1cm] (c) {$c$};
			\node[pl,right of = c, label=below:$r$] (r) {};
			\node[tr,right of = r] (d) {$d$};
			\path[->]
				(a) edge node[above,pos=.3]{$x$} (p)
				(p) edge node[above,pos=.3]{$x$} (b)
				(b) edge node[above,sloped]{$x$} (q)
				(q) edge node[above,sloped]{$x$} (c)
				(c) edge node[above,sloped]{$x$} (s)
				(s) edge node[above,sloped]{$x$} (b)
				(c) edge node[above,pos=.3]{$x$} (r)
				(r) edge node[above,pos=.3]{$x$} (d);  
			\end{tikzpicture}
        \caption{\tpnid $\compose{N}{M}$}
	\end{subfigure}
\caption{Although both $N$ and $M$ are \tjns, their composition is not}\label{fig:union-example}
\end{figure}
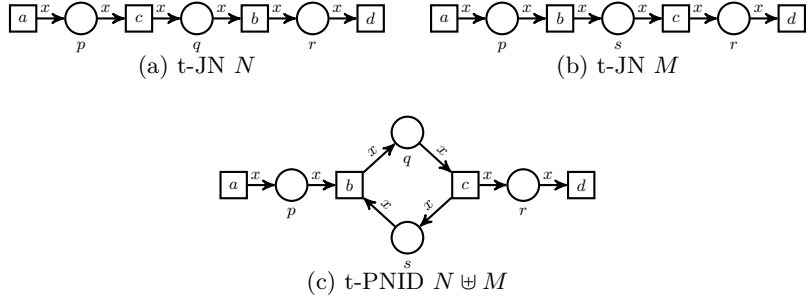

It is easy to see that the composition of two \tjns does not automatically result in a \tjn. Consider nets in \figref{union-example}. It is easy to see that both $N$ and $M$ can be obtained by applying \ref{itm:tjn_R2} from \defref{typed_jackson_net}. However, their composition 
cannot be reduced to a single transition by consecutively applying rules from \defref{typed_jackson_net}.

\begin{figure}
	\centering
	\begin{tikzpicture}[->,>=stealth',auto,node distance=20mm and 12mm,thick,  every node/.style={scale=0.8}]

\node[tr] (a) {$a$};
\node[pl,right of = a, label=above:$p_x$,yshift=15mm,green1] (px) {};
\node[pl,right of = a, label=below:$p_{xy}$] (p) {};
\node[pl,right of = a, label=below:$p_y$,yshift=-15mm,cadmiumorange] (py) {};
\node[tr, right of = px] (b) {$b$};
\node[tr, right of = py] (c) {$c$};
\node[pl,right of = b, label=above:$q_x$,green1] (qx) {};
\node[pl,right of = c, label=below:$q_y$,cadmiumorange] (qy) {};
\node[pl,right of = b, label=below:$q_{xy}$,yshift=-15mm] (q) {};
\node[tr, right of = q] (d) {$d$};

\path[->,green1]
(a) edge node[above]{$x$} (px)
(px) edge node[above]{$x$} (b)
(b) edge node[above]{$x$} (qx)
(qx) edge node[above]{$x$} (d)
(px) edge node[right, pos=.7]{$x$} (c)
(c) edge node[left, pos=.2]{$x$} (qx)
;  

\path[->,cadmiumorange]
(a) edge  node[below]{$y$} (py)
(py) edge node[below]{$y$} (c)
(py) edge node[right,pos=.7]{$y$} (b)
(b) edge node[left,pos=.2]{$y$} (qy)
(c) edge node[below]{$y$} (qy)
(qy) edge node[below]{$y$} (d)
;

\path[->,byzantium]
(a) edge node[above]{$xy$} (p)
(p) edge node[above,sloped, pos=.6]{$xy$} (b)
(p) edge node[below,sloped, pos=.6]{$xy$} (c)
(b) edge node[above,sloped,pos=.4]{$xy$} (q)
(c) edge node[below,sloped,pos=.4]{$xy$} (q) 
(q) edge node[above]{$xy$} (d)
;  
\end{tikzpicture}
	\caption{Composition of the projections on $\set{\lambda_1}$, $\set{\lambda_2}$ and $\set{\lambda_1,\lambda_2}$ on the \tjn $(a;[p,\tup{x,y}]; (b||c);[q,\tup{x,y}];d)$. Here, type assignments are as follows: 
	$\alpha(p_x)=\alpha(q_x)=\lambda_1$,
	$\alpha(p_y)=\alpha(q_y)=\lambda_2$ and 
	$\alpha(p)=\alpha(q)=\lambda_1\lambda_2$.}
	\label{fig:proj-union-example}
\end{figure}
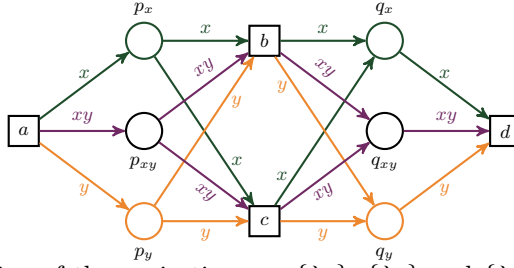

A more surprising observation is that composing type projections of a \tjn may not result in a \tjn. Take for example the net from Figure~\ref{fig:proj-union-example}. 
Both its projections on $\set{\lambda_1}$ and $\set{\lambda_2}$ are \tjns. 
However, bringing them together using the composition operator results in a \tpnid that is not \tjn: indeed, since the ``copies'' of place $p$ appear in three places, and all such copies have same pre- and post-sets (and only differ by their respective types), it is impossible to apply identifier elimination rule \emph{R6} from \defref{typed_jackson_net}.

As one may observe from the above example, 
the only difference between $[p_{xy},\tup{\lambda_1,\lambda_2}]$ and its copies $p_x$ and $p_y$ is in their respective types, whereas the identifiers carried by $p_x$ and $p_y$ are always contained in $p_{xy}$, and thus both $p_x$ and $p_y$ can be seen as subsidiary with respect to $p_{xy}$. 
We formalize this observation using the notion of \emph{minor places}:  
a place $p$ is minor to some place $q$ if both $p$ and $q$ have identical pre- and post-sets, and the type of $q$ subsumes the one of $p$.

\begin{definition}[Minor places] \deflabel{minor_places}
Let $N = (P_N, T_N, F_N, \alpha, \beta)$ be a \tpnid.
A place $p\in P$ is \emph{minor to} a place $q\in P$ iff the following holds:
\begin{compactitem}
\item  $\pre{p}=\pre{q}$, $\post{p}=\post{q}$ and $\alpha(p)\subset \alpha(q)$; 
\item  $\beta((t,p))=\restr{\beta((t,q))}{\type^{-1}(\alpha(p))}$, for each $t\in\pre{p}$;
\item  $\beta((p,t))=\restr{\beta((q,t))}{\type^{-1}(\alpha(p))}$, for each $t\in\post{p}$.
\end{compactitem}
\end{definition}

We show next that minor places can be added or removed without altering the overall behavior of the net.

\begin{lemma}
\label{lemma:minors}
Let $N = (P, T, F, \alpha, \beta)$ be a \tpnid with initial marking $m_0$ s.t. 
$m_0(p)=m_0(q)=\emptyset$, for $p,q\in P$, where $p$ is minor to $q$.
Let $N'=(P\setminus\set{p}, T, F\setminus(\set{(p,t)|t\in\post{p}}\cup\set{(t,p)|t\in\pre{p}}),\alpha,\beta)$ be a \tpnid obtained by eliminating from $N$ place $p$ .
Then $\transitionsystem{N,m_0}\sim^r \transitionsystem{N',m_0}$.
\end{lemma}
\begin{proof}(sketch)
It is enough to define a relation $Q\subseteq \pnreachable{N}{m_0}\times \pnreachable{N'}{m_0}$ s.t. $(m,m')\in Q$ iff $m(r)=m'(r)$, for $r\in P\setminus\set{p}$, and $m(p)(\cname{id})=m'(q)(\cname{id})$, for all $\cname{id}\in\colset(\place)$, and $|m(p)|=|m'(q)|$.
Then the lemma statement directly follows from the firing rule of \tpnids and that pre- and post-sets of $p$ and $q$ coincide.
\end{proof}

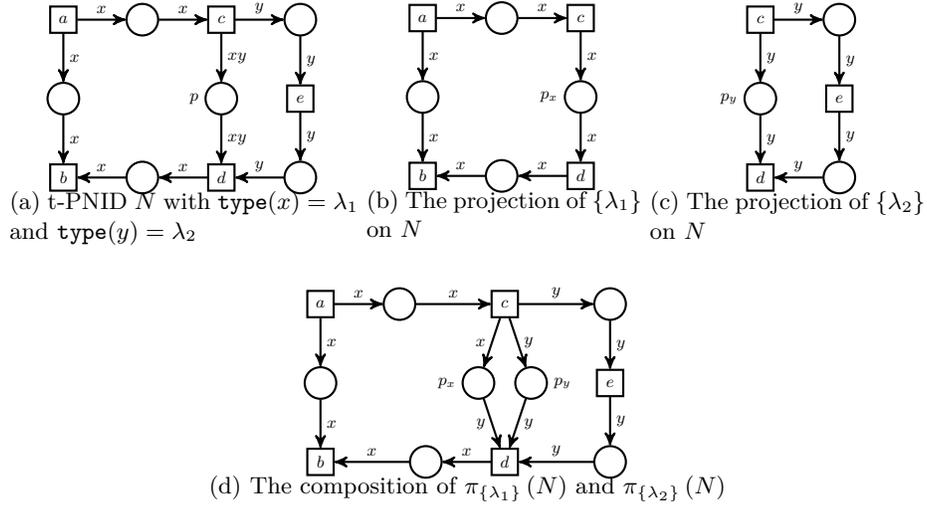
\begin{figure}[t!]
  \centering
	\begin{subfigure}{.38\textwidth}
		\centering
		 	\begin{tikzpicture}[->,>=stealth',auto,x=10mm,y=1cm,node distance=15mm and 3mm,thick,  every node/.style={scale=0.7}]
			\node[tr] (a) {$a$};
			\node[pl,right of = a] (p1) {};
			\node[tr,right of = p1] (c) {$c$};			
			\node[pl,right of = c] (p5){};
			\node[tr,below of = p5] (e){$e$};
			 \node[pl,below of = e] (p6){};
			\node[pl,below of = a] (p2) {};
			\node[tr,below of = p2] (b){$b$};
			\node[pl,below of = c,label=left:$p$] (p3) {};
			\node[tr,below of = p3] (d) {$d$};
			\node[pl, left of = d] (p4){};
			\path[->]
				(a) edge node[above]{$x$} (p1)
				(a) edge node[right]{$x$} (p2)
				(p1) edge node[above]{$x$} (c)
				(p2) edge node[right]{$x$} (b)
				(p4) edge node[above]{$x$} (b)
				(d) edge node[above]{$x$}  (p4)
				(c) edge node[right]{$xy$} (p3)
				(p3) edge node[right]{$xy$}  (d)
				(c) edge node[above]{$y$}  (p5)
				(p5) edge node[right]{$y$} (e)
				(e) edge node[right]{$y$} (p6)
				(p6) edge node[above]{$y$} (d);  
			\end{tikzpicture}
        \caption{\tpnid $N$ with $\type(x)=\lambda_1$ and $\type(y)=\lambda_2$\label{fig:singleton-tjn-n}}
	\end{subfigure}
	\begin{subfigure}{.3\textwidth}
		\centering
		\begin{tikzpicture}[->,>=stealth',auto,x=10mm,y=1cm,node distance=15mm and 3mm,thick,  every node/.style={scale=0.7}]
			\node[tr] (a) {$a$};
			\node[pl,right of = a] (p1) {};
			\node[tr,right of = p1] (c) {$c$};			
			\node[pl,below of = a] (p2) {};
			\node[tr,below of = p2] (b){$b$};
			\node[pl,below of = c,label=left:$p_x$] (p3) {};
			\node[tr,below of = p3] (d) {$d$};
			\node[pl, left of = d] (p4){};
			\path[->]
				(a) edge node[above]{$x$} (p1)
				(a) edge node[right]{$x$} (p2)
				(p1) edge node[above]{$x$} (c)
				(p2) edge node[right]{$x$} (b)
				(p4) edge node[above]{$x$} (b)
				(d) edge node[above]{$x$}  (p4)
				(c) edge node[right]{$x$} (p3)
				(p3) edge node[right]{$x$}  (d);
			\end{tikzpicture}
        \caption{The projection of $\set{\lambda_1}$ on $N$ \label{fig:singleton-tjn-m}}
	\end{subfigure}
	\begin{subfigure}{.3\textwidth}
		\centering
		\begin{tikzpicture}[->,>=stealth',auto,x=10mm,y=1cm,node distance=15mm and 3mm,thick,  every node/.style={scale=0.7}]
			\node[tr] (c) {$c$};			
			\node[pl,right of = c] (p5){};
			\node[tr,below of = p5] (e){$e$};
			 \node[pl,below of = e] (p6){};
			\node[pl,below of = c,label=left:$p_y$] (p3) {};
			\node[tr,below of = p3] (d) {$d$};
			\path[->]
				(c) edge node[right]{$y$} (p3)
				(p3) edge node[right]{$y$}  (d)
				(c) edge node[above]{$y$}  (p5)
				(p5) edge node[right]{$y$} (e)
				(e) edge node[right]{$y$} (p6)
				(p6) edge node[above]{$y$} (d);  
			\end{tikzpicture}
        \caption{The projection of $\set{\lambda_2}$ on $N$ \label{fig:singleton-tjn-k}}
	\end{subfigure}
	\begin{subfigure}{.7\textwidth}
		\centering
		 	\begin{tikzpicture}[->,>=stealth',auto,x=10mm,y=1cm,node distance=15mm and 3mm,thick,  every node/.style={scale=0.7}]
			\node[tr] (a) {$a$};
			\node[pl,right of = a] (p1) {};
			\node[tr,right of = p1,xshift=5mm] (c) {$c$};			
			\node[pl,right of = c,xshift=5mm] (p5){};
			\node[tr,below of = p5] (e){$e$};
			 \node[pl,below of = e] (p6){};
			\node[pl,below of = a] (p2) {};
			\node[tr,below of = p2] (b){$b$};
			\node[pl,below of = c,label=left:$p_x$,xshift=-.5cm] (px) {};
			\node[pl,below of = c,label=right:$p_y$,xshift=.5cm] (py) {};

			\node[tr,below of = px,xshift=.5cm] (d) {$d$};
			\node[pl, left of = d] (p4){};
			\path[->]
				(a) edge node[above]{$x$} (p1)
				(a) edge node[right]{$x$} (p2)
				(p1) edge node[above]{$x$} (c)
				(p2) edge node[right]{$x$} (b)
				(p4) edge node[above]{$x$} (b)
				(d) edge node[above]{$x$}  (p4)
				(c) edge node[left]{$x$} (px)
				(px) edge node[left]{$y$}  (d)
				(c) edge node[right]{$y$} (py)
				(py) edge node[right]{$y$}  (d)
				(c) edge node[above]{$y$}  (p5)
				(p5) edge node[right]{$y$} (e)
				(e) edge node[right]{$y$} (p6)
				(p6) edge node[above]{$y$} (d);  
			\end{tikzpicture}
        \caption{The composition of  $\project{\set{\lambda_1}}{N}$ and $\project{\set{\lambda_2}}{N}$\label{fig:compose-singleton}}
	\end{subfigure}
\caption{\tpnid $N$ (\ref{fig:singleton-tjn-n}), its singleton projections  and their composition }\label{fig:singleton-example}
\end{figure}

Let us now address the reconstructability property.
In a nutshell, a net is reconstructable if composing all of its type projections 
returns the same net. 
This property is not that trivial to obtain.
For example, let us consider singleton projections (that is, projections $\project{\set{\lambda}}{N}$ obtained for each $\lambda\in\type_\Lambda(N)$)
of the net in \figref{singleton-example}.
It is easy to see that such projections ``ignore'' interactions between objects (or system components).
Thus, the composition of the singleton projections $\project{\set{\lambda_1}}{N}$ and $\project{\set{\lambda_2}}{N}$ from \figref{singleton-example} does not result in a model
that merges $p_x$ and $p_y$ in one place as the composition operator cannot recognize component interactions between such projections. 
This is reflected in \figref{compose-singleton}.

To be able to reconstruct the original model from its projections (or at least do it approximately well), 
one needs to consider a projection reflecting component interactions. In the case of the net from Figure~\ref{fig:singleton-tjn-n}, its non-singleton projection $\project{\set{\lambda_1,\lambda_2}}{N}$ is depicted in Figure~\ref{fig:interactions-n}.
Now, using this projection we can obtain a composition (see Figure~\ref{fig:compose-full}) that closely resembles $N$. 
Notice that, in this composition, copies of the interaction place $p$ appear three times as places $p_x$, $p_y$ and $p_{xy}$, respectively.
It is also easy to see that places $p_x$ and $p_y$ are minor to $p_{xy}$, and 
$\alpha(p)=\alpha(p_{xy})$ witnesses that $\project{\set{\lambda_1,\lambda_2}}{N}$ is the maximal projection defined over types of $N$ s.t. the correct type of $p$ is ``reconstructed''.
This leads us to the following result stipulating the reconstructability property of typed Jackson nets.

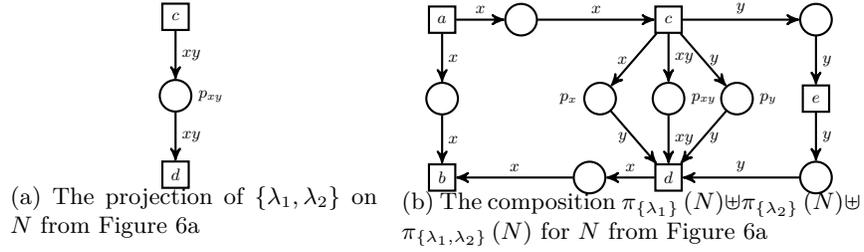
\begin{figure}[t!]
  \centering
	\begin{subfigure}{.4\textwidth}
		\centering
		\begin{tikzpicture}[->,>=stealth',auto,x=10mm,y=1cm,node distance=15mm and 3mm,thick,  every node/.style={scale=0.7}]
			\node[tr] (c) {$c$};			
			 \node[pl,below of = c,label=right:$p_{xy}$] (p){};
			\node[tr,below of = p] (d) {$d$};
			\path[->]
				(c) edge node[right]{$xy$} (p)
				(p) edge node[right]{$xy$} (d);  
			\end{tikzpicture}
        \caption{The projection of $\set{\lambda_1,\lambda_2}$ on $N$ from Figure~\ref{fig:singleton-tjn-n} \label{fig:interactions-n}}
	\end{subfigure}\quad
	\begin{subfigure}{.5\textwidth}
		\centering
		 	\begin{tikzpicture}[->,>=stealth',auto,x=10mm,y=1cm,node distance=15mm and 3mm,thick,  every node/.style={scale=0.7}]
			\node[tr] (a) {$a$};
			\node[pl,right of = a] (p1) {};
			\node[tr,right of = p1,xshift=13mm] (c) {$c$};			
			\node[pl,right of = c,xshift=13mm] (p5){};
			\node[tr,below of = p5] (e){$e$};
			 \node[pl,below of = e] (p6){};
			\node[pl,below of = a] (p2) {};
			\node[tr,below of = p2] (b){$b$};
			\node[pl,below of = c,label=left:$p_x$,xshift=-13mm] (px) {};
			\node[pl,below of = c,label=right:$p_y$,xshift=13mm] (py) {};
			\node[pl,below of = c,label=right:$p_{xy}$] (p) {};

			\node[tr,below of = px,xshift=13mm] (d) {$d$};
			\node[pl, left of = d] (p4){};
			\path[->]
				(a) edge node[above]{$x$} (p1)
				(a) edge node[right]{$x$} (p2)
				(p1) edge node[above]{$x$} (c)
				(p2) edge node[right]{$x$} (b)
				(p4) edge node[above]{$x$} (b)
				(d) edge node[above]{$x$}  (p4)
				(c) edge node[left]{$x$} (px)
				(px) edge node[left]{$y$}  (d)
				(c) edge node[right]{$xy$} (p)
				(p) edge node[right]{$xy$}  (d)
				(c) edge node[right]{$y$} (py)
				(py) edge node[right]{$y$}  (d)
				(c) edge node[above]{$y$}  (p5)
				(p5) edge node[right]{$y$} (e)
				(e) edge node[right]{$y$} (p6)
				(p6) edge node[above]{$y$} (d);  
			\end{tikzpicture}
			  \caption{The composition  $\project{\set{\lambda_1}}{N}\composeOperator\project{\set{\lambda_2}}{N}\composeOperator\project{\set{\lambda_1,\lambda_2}}{N}$  for $N$ from Figure~\ref{fig:singleton-tjn-n} \label{fig:compose-full}}
	\end{subfigure}
\caption{Adding the projection $\project{\set{\lambda_1,\lambda_2}}{N}$ reflecting interactions to the composition results in the original net $N$ modulo places minor to $p$ (such as $p_x$ and $p_y$). }\label{fig:compose-interactions}
\end{figure}

\begin{theorem}\thmlabel{reconstructability}
Let $N = (P, T, F, \alpha, \beta)$ be a \tjn. 
Then $\transitionsystem{N,\emptyset}\sim^r \transitionsystem{N',\emptyset}$, where\\ $N'=\biguplus\limits_{\emptyset\subset\Upsilon\subseteq\type_\Lambda(N)}\project{\Upsilon}{N}$. 
\end{theorem}
\begin{proof} (sketch)
The proof immediately follows from the next observation.
Among all possible projections, for each place $p\in P$ there exists a projection $\project{\Upsilon}{N}$ such that $\alpha(p)=\Upsilon$. This also means that $\project{\Upsilon}{N}$ contains $p$ and that all other projections $\project{\Upsilon'}{N}$ with $\Upsilon'\subset\Upsilon$ will at most include the minors of $p$. 
Following \defref{composition}, it is easy to see that the composition of all the projections yields a \tjn identical to $N$ modulo additional place minors introduced by some of the projections. Showing that the obtained net is bisimilar to $N$ can be done by analogy with Lemma~\ref{lemma:minors}.
\end{proof}

Notice that the above result can be made stronger if all the additional minors (i.e., minors that were not present originally in $N$) are removed using reduction rules from \defref{typed_jackson_net}. For simplicity, given a \tpnid $N$ with the set of places $P$, we denote by $\lfloor P \rfloor$ the set of its minor places. 

\begin{corollary}\corlabel{reconstructability}
Let $N$ be a \tjn and $N'$ is as in \thmref{reconstructability}. 
Then $(N,\emptyset)\leftrightsquigarrow(N',\emptyset)$, if $\lfloor P \rfloor = \lfloor P' \rfloor$, where $P$ and $P'$ are respectively the sets of places of $N$ and $N'$.
\end{corollary}
The above result can be obtained by complementing the proof of \thmref{reconstructability} with a step that applies finitely many \tjn reduction rules to all the minor places that are in $N'$ and not in $N$.

	\section{A Framework for Rediscoverability} \seclabel{framework}
In the previous section, we showed that \tjns enjoy the reconstructability property: given a \tjn, a composition of \emph{all} its (proper) type projections yields a \tjn that is strongly bisimilar to the original one.\footnote{Such nets are also isomorphic if minor places of the composition are removed by consecutively applying the reduction rules from \defref{typed_jackson_net}.}

In this section, we propose a framework to rediscover systems of interacting processes that rely on this property.
The framework builds upon a divide and conquer strategy~\cite{TourPKS2022agentdiscovery}.
The first step of the approach is to divide the event logs over all possible projections.
For this, we translate the notion of event logs to event logs of interacting systems, and show that if these event logs are generated by a \tjn, projections on these event logs have a special property: the projected event log can be replayed by the projected net. 
In other words, there is no distinction between the projection on the event log, or that the projected net generated the event log. 
This observation forms the basis of the proposed framework for rediscoverability.
In the second step, we conquer the discoverability problem of the system of interacting processes by first discovering a model for each of the projections, and then composing these projections into the original system.
If the event log and discovery algorithm guarantee the defined properties, composition yields rediscoverability.

\subsection{Event Logs and Execution Traces}
In process discovery, an event log is represented as a (multi)set of sequences of events (called traces), where each sequence represents an execution history of a process instance. 
Traditional process discovery assumes the process to be a \wfnet. 
Consequently, each trace in an event log should correspond to a sequence of transition firings of the workflow net.
If this is the case, the event log is said to be generated by the \wfnet. 
We generalize this notion to marked Petri nets.

\begin{definition}[Event Log]
Given a set of transitions $T$, a set of traces $L \subseteq \seq{T}$ is called an \emph{event log}.
An event log $L$ \emph{is generated by} a marked Petri net $(N,m)$ if $\pnenabled{(N,m)}{\sigma}{}$ for all $\sigma \in L$, i.e., $L \subseteq \mathcal{L}(N,m_0)$.
\end{definition}

\begin{table}[t]
	\caption{Firing sequence for the \tpnid in \figref{runningexample}}\tbllabel{sequencerunningexample}
	\begin{minipage}{.3\textwidth}
		\begin{tabular}{l|c|c|c}
			\textbf{transition} & \textbf{x} & \textbf{y} & \textbf{z} \\\hline
			$A$      &     $p1$ &            &             \\
			$A$      &     $p2$ &            &             \\
			$T$      &          &            &   $c1$      \\
			$G$      &          &   $o1$     &   $c1$      \\
			$C$      &     $p1$ &            &             \\
			$E$      &     $p2$ &   $o1$     &             \\
		\end{tabular}
	\end{minipage}
	\hfill
	\begin{minipage}{.3\textwidth}
		\begin{tabular}{l|c|c|c}
			\textbf{transition} 
			& \textbf{x} & \textbf{y} & \textbf{z} \\\hline
			$T$      &          &            &   $c2$ \\
			$H$      &          &  $o1$      &        \\
			$L$      &          &  $o1$      &        \\ 
			$J$      &          &  $o1$      &        \\ 
			$B$      &  $p2$    &            &        \\ 
			$O$      &          &  $o1$      &        \\
		\end{tabular}
	\end{minipage}
	\hfill
	\begin{minipage}{.3\textwidth}
		\begin{tabular}{l|c|c|c}
			\textbf{transition} & \textbf{x} & \textbf{y} & \textbf{z} \\\hline
			$D$      &   $p1$   &            &        \\ 
			$V$      &          &            &$c2$    \\ 
			$K$      &          &   $o1$     &        \\ 
			$Z$      &          &   $o1$     & $c1$   \\ 
			$V$      &          &            & $c1$   \\ 
			$B$      &  $p1$    &            &        \\
		\end{tabular}
	\end{minipage}
\end{table}

Each sequence in a single process event log assumes to start from the initial marking of the \wfnet. 
A marked \tpnid, instead, represents a continuously executing system, for which, given a concrete identifier, exists a single observable execution that can be recorder in an event log. 
Thus, event logs are partial observations of a larger execution within the system: an event log for a certain type captures only the relevant events that contain identifiers of that type, and stores these in order of their execution.
Since each transition firing consists of a transition and a binding,  a \tpnid firing sequence induces an event log for each set of types $\Upsilon$. 
Intuitively, this induced event log is constructed by a filtering process. 
For each possible identifier vector for $\Upsilon$ we keep a firing sequence. 
Each transition firing is inspected, and if its binding satisfies an identifier vector of $\Upsilon$, it is added to the corresponding sequence.

\begin{definition}[Induced Event Log]
	Let $(N,m_0)$ be a marked \tpnid. Given a non-empty set of types $\Upsilon \subseteq \type_{\Lambda}(N)$, the \emph{$\Upsilon$-induced event log} of a firing sequence $\eta \in \L(N,m_0)$ is defined by:
$		\mathit{Log}_\Upsilon(\eta) = \{ \proj{\eta}{i} \mid i \in (\Id(\eta) \intersect I(\Upsilon))^{\length{\Upsilon}} \} 
$,
	where $\proj{\eta}{i}$ is inductively defined by
	\begin{inparaenum}[\it (1)]
		\item $\proj{\emptysequence}{i} = \emptysequence$, 
		\item $\proj{(\sequence{(t,\psi)}\concat \eta )}{i} = \sequence{(t,\psi)}\concat\proj{\eta}{i}$ if $\setsuppp{i}\subseteq \rng{\psi}$, and
		\item $\proj{(\sequence{(t,\psi)}\concat \eta )}{i} = \proj{\eta}{i}$ otherwise.
		\end{inparaenum}
\end{definition}

Different event logs can be induced from a firing sequence.
Consider, for example, the firing sequence of the net from \figref{runningexample} represented as table in \tblref{sequencerunningexample}.
As we cannot deduce the types for each of the variables from the firing sequences in \tblref{sequencerunningexample}, we assume that there is a bijection between variables and types, i.e., that each variable is uniquely identified by its type, and vice-versa.
Like that, we can create an induced log for each variable, as the type and variable name are interchangeable. 
For example, the $x$-induced event log is $\mathit{Log}_{\{x\}} = \{\sequence{A,E,B},\sequence{A,C,D,B}\}$, and the $z$-induced event log is $\mathit{Log}_{\{z\}} = \{\sequence{T,G,Z,V},\sequence{T,V}\}$.
Similarly, event logs can be also induced for combinations of types.
In this example, the only non-empty induced event logs on combined types are $\mathit{Log}_{\{y,z\}} = \{\sequence{G,Z}\}$ and $\mathit{Log}_{\{x,y\}} = \{\sequence{E}\}$.

As the firing sequence in \tblref{sequencerunningexample} shows, 
transition firings (and thus also events) only show bindings of variables to identifiers.
For example, for firing $G$ with binding $y \mapsto o1$ and $z \mapsto c1$, it is not possible to derive the token types of the consumed and produced tokens directly from the table. 
Therefore, we make the following assumptions for process discovery on \tpnids:
\begin{enumerate}
	\item There are no ``black'' tokens: all places carry tokens with at least one type, and all types occur at most once in a place type, i.e., all places refer to at least one process instance.
	\item There is a bijection between variables and types, i.e., for each type exactly one variable is used.
	\item A G\"odel-like number $\mathscr G$ is used to order the types in place types, i.e., for any place $p$, we have $\mathscr G(\alpha(p)(i)) < \mathscr G(\alpha(p)(j))$ for $1 \leq i < j \leq \length{\alpha(p)}$ and $p \in P$.
\end{enumerate}

%
%

\subsection{Rediscoverability of Typed Jackson Nets}
Whereas traditional process discovery approaches relate events in an event log to a single object: the process instance, object-centric approaches can relate events to many objects~\cite{GhahfarokhiPBA21}.
Most object-centric process discovery algorithms (e.g., \cite{aalstB20_discovering,LuNWF15}) use a divide and conquer approach, where ``flattening'' is the default implementation to divide the event data in smaller event logs. 
The flattening operation creates a trace for each object in the data set, and combines the traces of objects of the same type in an event log. 
As we have shown in \secref{decomposability}, singleton projections, i.e., those just considering types in isolation, are insufficient to reconstruct the \tjn that induced the object-centric event log. 
A similar observation is made for object-centric process discovery (cf.~\cite{aalst19_divergence,aalstB20_discovering,adams22_extractingfeatures}): flattening the event data into event logs generates inaccurate models.
Instead, reconstructability can only be achieved if all possible combinations of types are considered. 
Hence, for a divide and conquer strategy, the divide step should involve all possible combinations of types, i.e., each interaction between processes requires their own event log.
In the remainder of this section, we show that if all combinations of types are considered, flattening is possible, and traditional process discovery algorithms can be used to rediscover a system of interacting processes.

For a system of interacting processes, we consider execution traces, i.e., a firing sequence from the initial marking. 
Like that, event logs for specific types or combinations of types are induced from the firing sequence. 
The projection of the system on a type or combinations of types, results again in a \tjn. 
Similarly, if we project a firing sequence of a \tjn $N$ on a set of types~$\Upsilon$, then this projection is a firing sequence of the $\Upsilon$-projection on~$N$.
The property follows directly from the result that \tjn $N$ is weakly simulated by its $\Upsilon$-projection.

\begin{lemma}
Let $N$ be a \tjn, and let $\Upsilon \subseteq \type_{\Lambda}(N)$. Then 
$\hide{U}{\transitionsystem{N,\emptybag}}\preccurlyeq^r \transitionsystem{\project{\Upsilon}{N},\emptybag}$, with $U = T_N \setminus T_\Upsilon$.
\end{lemma}
\begin{proof} (sketch)
Let $N_\Upsilon = \proj{\Upsilon}{N} = (P_\Upsilon, T_\Upsilon, F_\Upsilon,\alpha_\Upsilon, \beta_\Upsilon)$.
We can define a relation $Q \subseteq \markings{N} \times \markings{\project{\Upsilon}{N}}$ s.t. $Q(m)(p)(\proj{a}{I(\Upsilon)}) = m(p)(a)$ if $p \in P_\Upsilon$ and $Q(m)(p) = m(p)$ otherwise. 
The rooted weak bisimulation of $Q$ follows directly from the firing rule of \tpnids.
\end{proof}

As the lemma shows, projecting a firing sequence yields a firing sequence for the projected net.
A direct consequence of the simulation relation is that, no matter whether we induce an event log from a firing sequence on the original net, or induce it from the projected firing sequence, the resulting event logs are the same.

\begin{corollary}
Let $(N,m_0)$ be a marked \tpnid. Given a set of types $\Upsilon \subseteq \type_{\Lambda}(N)$. 
Then $\mathit{Log}_\Upsilon(\eta) = \mathit{Log}_\Upsilon(\project{\Upsilon}{\eta})$.
\end{corollary}

Hence, it is not possible to observe whether an induced event log stems from the original model, or from its projection. 
Note that the projection may exhibit more behavior, so the reverse does not hold. 
In general, not any induced event log from the projection can be induced from the original model. 

In general, a projection does not need to be an atomic \tjn (that is, a \tjn that can be reduced by applying rules from \defref{typed_jackson_net} to a single transition). 
However, if the projection is atomic, then its structure is a transition-bordered \wfnet: 
a \wfnet that, instead of having source and sink places, has a set of start and finish transitions, such that pre-sets (resp., post-sets) of start (resp., finish) transitions are empty.
The closure of a transition-bordered \wfnet is constructed by adding a new source place $i$ so that each start transition consumes from $i$, and a new sink place $f$ so that each finish transition produces in $f$.

\begin{figure}[t]
	\centering
	\begin{tikzpicture}[->,>=stealth',auto,x=10mm,y=1cm,node distance=15mm and 3mm,thick,  every node/.style={scale=.9}]
		\node[tr, label=center:$M$] (M) {};
		\node[tr, right of = M, label=center:$M_1$, yshift=0.5cm, xshift=1.75cm]  (M1) {}; 
		\node[below of = M1,rotate=90,xshift=10mm] (d1) {$\cdots$};
		\node[tr, right of = M, label=center:$M_n$, yshift=-0.5cm, xshift=1.75cm]  (Mn) {}; 
		\node[tr, right of = M, label=center:$M'$, xshift=5cm]  (M') {};
		
		\node[tr, below of = M, label=center:$L$, yshift=-0.5cm]    (L) {}; 
		\node[tr, right of = L, label=center:$L_1$, yshift=0.5cm]   (L1) {}; 
		\node[below of = L1,rotate=90,xshift=10mm] (d2) {$\cdots$};
		\node[tr, right of = L, label=center:$L_n$, yshift=-0.5cm]  (Ln) {};
		
		\node[tr, right of = L1, label=center:$D_1$, xshift=2cm] (D1) {};
		\node[below of = D1,rotate=90,xshift=10mm] (d3) {$\cdots$};
		\node[tr, right of = Ln, label=center:$D_n$, xshift=2cm] (Dn) {};
		\node[tr, below of = M', label=center:$D'$, yshift=-0.5cm]  (D') {};
		
		\node[right of = M, xshift=-0.8cm] () {$\pi$};
		\node[right of = L, xshift=-0.8cm] () {$\pi$};
		
		\node[left of = M', xshift=0.8cm] () {$\uplus$};
		\node[left of = D',xshift=0.8cm] () {$\uplus$};	
		
		\node[below of = Ln,yshift=1.2cm] (a) {};
		
		
		\path[->, thin]
		($(M.east)+(0,.8mm)$) edge (M1)
		($(M.east)-(0,.8mm)$) edge (Mn)
		
		(L) edge (L1)
		(L) edge (Ln)
		
		(M1) edge ($(M'.west)+(0,.8mm)$)
		(Mn) edge ($(M'.west)-(0,.8mm)$)
		
		(D1) edge (D')
		(Dn) edge (D')
		;
		\path[->, draw=blue-violet]
		(M) edge (L)
		(M1) edge (L1)
		(Mn) edge (Ln)
		;
		\path[<->, draw=cadmiumorange]
		(M) edge[bend left=30] node[above,cadmiumorange,yshift=-.5mm]{$\sim^r$} (M')
		;
		\path[->, draw=red]
		(L1) edge node[above,red]{$\overline{\mathit{disc}}$} (D1)
		(Ln) edge node[above,red]{$\overline{\mathit{disc}}$} (Dn)
		;
		\path[<->, draw=green1]
		(M1) edge (D1)
		(Mn) edge (Dn)
		(M') edge (D')
		;
	\end{tikzpicture}

	\caption{Framework for rediscoverability of typed Jackson Nets. Model $M$ generates an event log $L$. Log projections $L_1 \ldots L_n$ are generated from projected nets $M_1 \ldots M_n$. 
		Discovery algorithm $\mathit{disc}$ results in nets $D_1 \ldots D_n$, isomorphic to $M_1 \ldots M_n$, 
		which can be composed in $D'$.
		$D'$ is isomorphic to $M'$ and thus to $M$.
	}
	\figlabel{framework}
\end{figure}
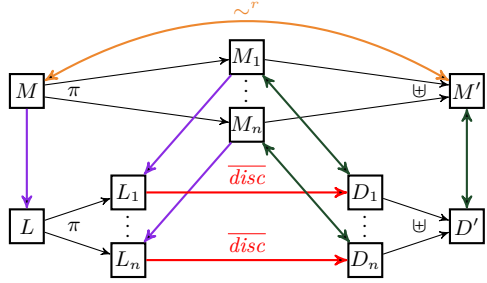

\begin{lemma}
	Let $N$ be a \tjn and $\project{\Upsilon}{N} = (P_\Upsilon, T_\Upsilon, F_\Upsilon, \alpha_\Upsilon, \beta_\Upsilon)$ for some $\Upsilon \subseteq \type_{\Lambda}(N)$ such that $\project{\Upsilon}{N}$ is atomic. 
	Let $\eta \in \mathcal{L}(N, \emptybag)$ be a firing sequence. 
	Then $\mathit{Log}_\Upsilon(\eta)$ is generated by $(N_\Upsilon, \emptybag)$ with $N_\Upsilon = (P_\Upsilon \union \{i, f\},T_\Upsilon,F_\Upsilon \{ (i,t) 
	\mid \pre{t}=\emptyset \} \union \{(t,f) \mid \post{t} = \emptyset\})$.
\end{lemma}
\begin{proof} (sketch)
Let $\sigma \in \mathit{Log}_\Upsilon(\eta)$.
By construction, each firing sequence in $\mathit{Log}_\Upsilon(\eta)$ has some corresponding identifier vector that generated the sequence.
Assume $\vec\upsilon \in \idset^{|\Upsilon|}$ is such a vector for $\sigma$. 

Observe that for any transition $t \in T$ if $\pre{t} = \emptyset$, $\newvar{t} \intersect \Upsilon \neq \emptyset$, and similarly, if $\post{t} = \emptyset$, $\delvar{t} \intersect \Upsilon \neq \emptyset$. As $N$ is identifier sound, only $\pre{\sigma(1)} = \emptyset$ and $\post{\sigma(\length{\sigma})} = \emptyset$.
Define relation $R = \{(M,m) \mid \forall p \in P : M(p)(\upsilon)= m(p) \}$ and  $U = \{ (t,\psi) \mid \upsilon \not\subseteq \rng{\psi} \}$, i.e., $U$ contains all transitions that do not belong to $\sigma$. Then $R$ is a weak simulation, i.e., $\hide{U}{\transitionsystem{N,\emptybag}} \preccurlyeq^r_R \transitionsystem{N_\Upsilon,\emptybag}$ and thus $\pnenabled{(N_\Upsilon,\emptybag)}{\sigma}$.
\end{proof}

Given a set of types $\Upsilon$, if its projection is atomic, the projection can be transformed into a workflow net, and for any firing sequence of the original net, this \wfnet can generate the $\Upsilon$-induced event log.
Suppose we have a discovery algorithm $\mathit{disc}$ that can rediscover models, i.e., given an event log $L$ that was generated by some model $M$, then $\mathit{disc}$ returns the original model. 
Rediscoverability of an algorithm requires some property $P_{\mathit{disc}}(M)$ on the generating model $M$, and some property $Q_{\mathit{disc}}(L, M)$ on the quality of event log $L$ with respect to the generating model $M$. 
In other words, $P(M)$ and $Q(L,M)$ are premises to conclude rediscoverability for discovery algorithm $\mathit{disc}$.
For example, $\alpha$-miner~\cite{AalstWM04} requires for $P(M)$ that model $M$ is well-structured, and for $Q(L,M)$ that event log $L$ is directly-follows complete with respect to model $M$. 
Similarly, Inductive Miner~\cite{Leemans2013} requires the generating model $M$ to be a process tree  without silent actions or self-loops ($P(M)$), and that event log $L$ is directly-follows complete with respect to the original model $M$ ($Q(L,M)$).

\begin{definition}[Rediscovery]\deflabel{rediscovery}
An algorithm $\mathit{disc}$ can \emph{rediscover} \wfnet $W=(P,T,F,in,out)$ from event log $L \subseteq \seq{T}$ if $P_{\mathit{disc}}(W)$ and $Q_{\mathit{disc}}(L,W)$ imply $\mathit{disc}(L) \leftrightsquigarrow W$.
\end{definition}

Thus, suppose there exists a discovery algorithm $\mathit{disc}$ that is -- under conditions $P$ and $Q$ -- able to reconstruct a workflow model given an event log. 
In other words, given an event log $L$ generated by some model $M$, $\mathit{disc}$ returns a model that is isomorphic to the generating model. 
Now, suppose we have a firing sequence $\eta$ of some \tjn $N$, and some projection $\Upsilon$. Then, if $P(\project{\Upsilon}{N})$, and $Q(\mathit{Log}_\Upsilon(\eta),\project{\Upsilon}{N})$, then $\mathit{disc}$ returns a model that is isomorphic to the closure of $\project{\Upsilon}{N}$, as $\mathit{disc}$ only returns \wfnets. 
With $\overline{\mathit{disc}}$ we denote the model where the source and sink places are removed, i.e., $\overline{\mathit{disc}} \leftrightsquigarrow \project{\Upsilon}{N}$.
Then, as shown in \figref{framework}, if we discover for every possible combination of types, i.e., the subset-closed set of all type combinations, 
a model that is isomorphic to the type-projected model, then the composition results in a model that is bisimilar to the original model.

\begin{theorem}[Rediscoverability of typed Jackson Nets]
Let $N$ be a \tjn, and let $\eta \in \mathcal{L}(N,\emptyset)$ without minor places. 
Let $\mathit{disc}$ be a discovery algorithm with properties $P$ and $Q$ that satisfy \defref{rediscovery}. 
If for all $\emptyset \subset \Upsilon \subseteq \type_{\Lambda}(N)$ 
	the $\Upsilon$-projection is atomic and 
	satisfies conditions $P(\project{\Upsilon}{N})$ and $Q(\mathit{Log}_{\Upsilon}(\eta)), \project{\Upsilon}{N})$,
then $\transitionsystem{N,\emptyset} 
\leftrightsquigarrow
\transitionsystem{N',\emptyset}$ with $N' = \biguplus_{\emptyset \subset \Upsilon \subseteq \type_{\Lambda}(N)} \overline{\mathit{disc}}(\mathit{Log}_\Upsilon(\eta))$.

%
\end{theorem}
\begin{proof} (sketch)
Let $\emptyset \subset \Upsilon \subseteq \type_{\Lambda}(N)$ be a set of types in $N$. 
Since $P(\project{\Upsilon}{N})$ and $Q(\mathit{Log}_{\Upsilon}(\eta)), \project{\Upsilon}{N})$%
the closure of $\project{\Upsilon}{N}$ and $\mathit{disc}(\mathit{Log}_{\Upsilon}(\eta))$ are isomorphic.
From the closure, places $\inp$ and $\outp$ exist with $\pre{\inp} = \emptyset = \post{\outp{}}$. As the nets are isomorphic, we have $\proj{\Upsilon}{N} \leftrightsquigarrow \overline{\mathit{disc}}(\mathit{Log}_{\Upsilon}(\eta))$.
Combining the results gives
$\biguplus_{\emptyset \subset \Upsilon \subseteq \type_{\Lambda}(N)} \overline{\mathit{disc}}(\mathit{Log}_\Upsilon(\eta)) \leftrightsquigarrow
\biguplus_{\emptyset \subset \Upsilon \subseteq \type_{\Lambda}(N)} \project{\Upsilon}{N}$.
The statement then follows directly from \corref{reconstructability}.
\end{proof}

\section{Conclusion} 
\seclabel{discussion_conclusion}
In this paper, we studied typed Jackson Nets to model systems of interacting processes, a class of well-structured process models describing manipulations of object identifiers.
As we show, this class of nets has an important property of reconstructability.
In other words, the composition of the projections on all possible type combinations returns the model of the original system.
Ignoring the interactions between processes results in less accurate, or even wrong, models.
Similar problems occur in the discovery of systems of interacting processes, such as  object-centric process discovery, where event logs are flattened for each object.

This paper provides a formal foundation for the composition of block-structured nets, and uses this to develop a framework for the discovery of systems of interacting processes.
We link the notion of event logs used for process discovery to system executions, and show that it is not possible to observe whether an event log is generated by a system of interacting processes, or by a projection of the system.
These properties form the key ingredients of the framework.
We show under what conditions a process discovery algorithm (that guarantees rediscoverability) can be used to discover the individual processes and their interactions, and how these can be combined to rediscover a model of interacting processes that is bisimilar to the original system that generated the event logs.

Although typed Jackson Nets have less expressive power than formalisms like Object-centric Petri nets~\cite{aalstB20_discovering}, proclets~\cite{Fahland2019} or interacting artifacts~\cite{LuNWF15}, this paper shows the limitations and potential pitfalls of discovering interacting processes. 
This work aims to lay formal foundations for object-centric process discovery.
As a next step, we plan to implement the framework and tune our algorithms to discover useful models from industrial datasets.

\smallskip
\noindent
\textbf{Acknowledgements.}
Artem Polyvyanyy was in part supported by the Australian Research Council project DP220101516.

	\bibliographystyle{setup/splncs04}
	\bibliography{references}
\end{document}